\documentclass[journal]{IEEEtran}
\usepackage{amsthm}
\newtheorem{myDef}{Definition}
\newtheorem{theorem}{Theorem}

\usepackage[utf8]{inputenc} 
\usepackage[T1]{fontenc}    
\usepackage{hyperref}       
\usepackage{url}            
\usepackage{booktabs}       
\usepackage{amsfonts}       
\usepackage{nicefrac}       
\usepackage{microtype}      

\usepackage{xcolor}
\usepackage{bbding}
\usepackage{enumitem}
\usepackage{amsmath}
\usepackage{amssymb}
\usepackage{subfigure}
\usepackage{color}
\usepackage{graphicx}
\usepackage{algorithm,algpseudocode,algorithmicx}
\usepackage{amsmath,latexsym,amssymb,amsthm,array,amsfonts,algorithm,algpseudocode,booktabs,graphicx,subfigure,multirow,cuted,stfloats}
\usepackage{footmisc}
\usepackage{cite}
\usepackage{units}
\usepackage{tabularx}
\usepackage{multirow}
\usepackage{bm}
\usepackage{float}
\usepackage{soul}
\newlength\savewidth

\ifCLASSINFOpdf
\else
\fi
\begin{document}
%
\title{PIVODL: Privacy-preserving vertical federated learning over distributed labels}
%
%
%

\author{Hangyu~Zhu,
        Rui~Wang,
        Yaochu~Jin,~\IEEEmembership{Fellow,~IEEE},
        and~Kaitai~Liang,~\IEEEmembership{Member,~IEEE}
\thanks{Manuscript received xxx, 2021; revised xxx, 2021. (\textit{Corresponding authors: Yaochu Jin})}
\thanks{H. Zhu and Y. Jin are with the Department of Computer Science, University of Surrey, Guildford, Surrey GU2 7XH, UK. E-mail:\{hangyu.zhu; yaochu.jin\}@surrey.ac.uk.}
\thanks{R. Wang and K. Liang are with the Department of Intelligent Systems, Delft University of Technology, Delft 2628XE, the Netherlands. E-mail:\{R.Wang-8; Kaitai.Liang\}@tudelft.nl}}

%
%

\markboth{Journal of \LaTeX\ Class Files,~Vol.~xx, No.xx, August~2021}%
{Shell \MakeLowercase{\textit{et al.}}: Bare Demo of IEEEtran.cls for IEEE Journals}
%



\maketitle

\begin{abstract}
Federated learning (FL) is an emerging privacy preserving machine learning protocol that allows multiple devices to collaboratively train a shared global model without revealing their private local data. Non-parametric models like gradient boosting decision trees (GBDT) have been commonly used in FL for vertically partitioned data. However, all these studies assume that all the data labels are stored on only one client, which may be unrealistic for real-world applications. Therefore, in this work, we propose a secure vertical FL framework, named PIVODL, to train GBDT with data labels distributed on multiple devices. Both homomorphic encryption and differential privacy are adopted to prevent label information from being leaked through transmitted gradients and leaf values. Our experimental results show that both information leakage and model performance degradation of the proposed PIVODL are negligible.
\end{abstract}


%
\IEEEpeerreviewmaketitle

\section{Introduction}
%
%
%
%

\IEEEPARstart{D}{ata} privacy has become the main focus of attention in modern societies and the recently enacted General Data Protection Regulation (GDPR) prohibits users from wantonly sharing and exchanging their personal data. This may be a big barrier to model training, since standard centralized machine learning algorithms require to collect and store training data on one single cloud server. To tackle this issue, federated learning (FL) \cite{mcmahan2017communication} is proposed to enable multiple edge devices to collaboratively train a shared global model while keeping all the users' data on local devices. 

FL can be categorized into horizontal federated learning (HFL) and vertical federated learning (VFL) based on how data is partitioned \cite{10.1145/3298981}. HFL or instance-based FL represents the scenarios in which the users' training data share the same feature space but have different samples. A large amount of research work \cite{8241854,10.1145/3338501.3357370,10.1145/3338501.3357371,254465,zhu2020distributed} are dedicated on enhancing the security level of HFL, since recent studies \cite{10.1145/2810103.2813687,orekondy18gradient,NEURIPS2020_c4ede56b,li2019end} have shown that HFL protocol still suffer from potential risks of leaking local private data information. Secure multi-party  computation \cite{goldreich1998secure, 10.1145/3133956.3133982}, homomorphic encryption (HE) \cite{gentry2009fully} and differential privacy (DP) \cite{10.1007/978-3-540-79228-4_1} are three most common privacy preserving mechanisms that are theoretically and empirically proved to be effective for HFL. 

Compared to HFL, VFL is more likely to happen in the real-world applications. And the training data of participating clients in VFL have the same sample ID space but with different feature space. Privacy preservation is also a critical concern in VFL. Hardy \emph{et al.} \cite{hardy2017private} introduce a secure identity (ID) alignment framework to protect the data ID information in vertical federated logistic regression. Meanwhile, Nock \emph{et al.} \cite{nock2018entity} give a comprehensive discussion of the impact of ID entity resolution in VFL. In addition, Liu \emph{et al.} \cite{liu2020asymmetrically} point out the privacy concern of ID alignment in asymmetrical VFL. Yang \emph{et al.} introduce a simplified two-party vertical FL framework \cite{yang2019parallel} by removing the third party coordinator. Different from the aforementioned work for training \emph{parametric} models in VFL, Cheng \emph{et al.} \cite{cheng2021secureboost} first propose a secure XGBoost \cite{10.1145/2939672.2939785} decision tree system named SecureBoost in a setting of vertically partitioned data with the help of homomorphic encryption. Based on this work, Wu \emph{et al.} introduce a novel approach called Pivot \cite{10.14778/3407790.3407811} to ensure that the intermediate information is not disclosed during training. Tian \emph{et al.}\cite{tian2020federboost} design a Federboost scheme to train a GBDT over both HFL and VFL and vertical Federboost can satisfy the security requirements without any encryption operations.

However, current privacy preserving VFL systems are built under the assumption that all the data labels are stored only on one \emph{guest} or \emph{active} party, which is not realistic in many real-world applications. A typical VFL scenarios comes from two financial agents in the same region and these two agents provide different services but may have many common customers (residents in this region). It is more common that each agent owns parts of data labels like personal credit ratings, instead of only one agent contains all the label information. 

Therefore in this paper, we consider a more realistic setting of training XGBoost decision tree models in VFL in which each participating client holds parts of data labels that cannot be shared and exchanged with others during the training process. Compared to the standard vertical boosting tree system, constructing a secure vertical federated XGBoost system over distributed labels has the following challenges:
\begin{enumerate}
    \item Protocols like SecureBoost can perform gradient and Hessian summation on the guest client given the split for the tree nodes. However, in our assumption, data labels are distributed across connected clients and the summation cannot be directly operated if some labels of data samples within the split branch are missing. 
    
    \item XGBoost often needs to traverse all possible features to find the best split with the highest gain score. If each client is responsible for both splitting the data instances and computing the gain scores, the gradients and Hessian of data samples (contain label information) may run a high risk of being leaked. Even of homomorphic encryption is adopted in the learning system, the gradients can still be tracked by the differential attack (the difference between summations of gradient and Hessian for two adjacent feature splits).
    
    \item Training boosting trees require the intermediate leaf weight values to sequentially update the label predictions. These predictions can only be kept on guest clients to prevent other \emph{host} clients that do not contain any labels from deducing the label information. However, if the labels of data samples on split leaf nodes come from multiple clients, any client related to this node may track and guess the real labels from other clients through the leaf weight.
\end{enumerate}

To tackle the above challenges, we propose a novel privacy-preserving vertical federated learning system over distributed labels (PIVODL). PIVODL allows multiple clients to jointly construct a XGBoost tree model without disclosing feature or label information of the training data, given that the data labels are distributed across each of the participating clients in VFL. Specifically, additive HE and partial DP are adopted in the PIVODL framework to make sure that private label information will not be revealed or deduced during node split and prediction update. The contributions of the present work can be summarized as follows:
\begin{enumerate}
    \item We are the first to consider training XGBoost decision trees in VFL, in which the data labels are distributed over multiple data owners. The potential risk of privacy leakage under this condition is discussed in detail.
    \item A novel secure protocol is proposed by setting \emph{source} clients and \emph{split} clients for node split. A source client contains all information about split features, binning thresholds and split data indices, but is not used to calculate the gradients and Hessian sum of split branches if some branch data labels are unavailable. By contrast, a split client is used to compute impurity gains and leaf weights, but has no idea about the features and data information. By combining the source clients with the split clients during tree node split in VFL, we can effectively defend differential attacks and prevent intermediate gradients and Hessian values from being leaked.
    \item The calculated leaf node weights always need to be sent to other clients for label prediction updates. However as mentioned before, the source clients own all the split data indices and they know parts of prediction update for data instances of other clients. Therefore, an extra partial DP scheme is introduced by adding Gaussian noise to the leaf weights before sending them to the source clients.
    \item Empirical experiments are performed to compare the training time and model performance for both classification and regression tasks. In addition, we conduct an attack inference to deduce the data labels through intermediate information. Our results confirm that the proposed PIVODL system can effectively protect users' data privacy with negligible model performance degradation.
\end{enumerate}

\section{PRELIMINARIES}

\subsection{Vertical federated learning}
Different HFL \cite{yang2019federated}, where each client owns all features of the training set but differs in data samples, VFL \cite{liu2019communication} mainly focuses on the scenario where features are distributed among different clients. We denote $\mathcal{X}$ as the feature space, $\mathcal{Y}$ as the label space and $\mathcal{I}$ as the data IDs, the standard VFL can be defined as:

\begin{myDef}
VFL: Given a training set with $m$ data points distributed across $n$ clients, each client $j$ has the data feature $\mathcal{X}^{j} = \{X_{1}^{j}, \cdots, X_{m}^{j} \}$, labels $\mathcal{Y}^{j} = \{ y_{1}^{j}, \cdots, y_{m}^{j} \}$, and sample ids $\mathcal{I}^{j} =\{ I_{1}^{j}, \cdots, I_{m}^{j} \} $ where $j \in \{1,\cdots,n\}$. For any two different clients $j, j^{'}$, they satisfy:
\begin{equation}
\mathcal{X}^{j} \neq \mathcal{X}^{j^{'}}, \mathcal{Y}^{j} = \mathcal{Y}^{j^{'}}, \mathcal{I}^{j} = \mathcal{I}^{j^{'}}, j \neq j^{'}
\end{equation}
\end{myDef}

It can be seen that each connected client $j$ in VFL shares the same data sample ids $\mathcal{I}^{j}$ with the same corresponding labels $\mathcal{Y}^{j}$, but different clients may hold $\mathcal{X}^{j}$ sampled from different data feature space. And there are two kinds of clients in the standard VFL: one is the \emph{guest} or \emph{active} client, the other is the \emph{host} or \emph{passive} client.

\begin{myDef}
Guest Client: it holds both data features $\mathcal{X}$ and labels $\mathcal{Y}$ and is responsible for calculating the loss function, gradients, Hessians and leaf values for the corresponding data samples.
\end{myDef}

\begin{myDef}
Host Client: it only holds data features $\mathcal{X}$ and is responsible for aggregating the encrypted gradients and Hessians in one bucket.
\end{myDef}

In general, the guest client and host client can be seen as parameter server and node worker defined in distributed data parallelism \cite{isard2007dryad,isard2009distributed,fetterly2009dryadlinq}. For training parametric models like logistic regression \cite{hosmer2013applied} in Algorithm \ref{alg:vfllogistic}, unlike the parameter server used in HFL and data parallelism for model parameters aggregation, the guest client in VFL aggregates the model logits $z_{b}^{k}$ from the host clients and construct the training loss function $L(y_{b},\hat{y}_{b})$ with local data labels $\hat{y}_{b}$. Therefore, the derivative of the loss with respect to each received logits $\frac{\partial L}{\partial z_{b}}$ can only be computed on the guest client. As a result, only logits $z_{b}^{k}$ and derivatives $\frac{\partial L}{\partial z_{b}}$ need to be communicated between the guest client and host clients, and the communication costs are only dependent on the number of data samples $\mathcal{X}$.

\begin{algorithm}[htbp]\footnotesize{
\caption{VFL for logistic regression} 
\algblock{Begin}{End}
\label{alg:vfllogistic}
\begin{algorithmic}[1]
\State Training data $\mathcal{X}=\left \{ \mathcal{X}^{1},\mathcal{X}^{2},...\mathcal{X}^{K} \right \}$ on $K$ clients
\State Initialize the local model ${\theta_{0}^{k}}$, $k \in (1, K)$ \\
\For {each communication round $ t = 1,2,...T\ $}
\For{batch data $\mathcal{X}_{b}^{k} \in (\mathcal{X}_{1}^{k}, \mathcal{X}_{2}^{k},...\mathcal{X}_{B}^{k})$}
\For {each \textbf{Client} $k = 1,2,...K$ in parallel}
\State Compute $z_{b}^{k}=\mathcal{X}_{b}^{k} \theta_{t}^{k}$ and send it to  \textit{guest client}
\EndFor
\State Compute $\hat{y_{b}}=a(\sum_{k=1}^{K}z_{b}^{k})$ and $L(y_{b},\hat{y}_{b})$ on \textit{guest client}
\State Compute each $\frac{\partial L}{\partial z_{b}}$ on the \textit{guest client} and send it to the corresponding \emph{host client}
\For{each \textbf{Client} $k = 1,...K$ in parallel}
\State $\theta_{t}^{k} \leftarrow \theta_{t}^{k}-\eta \frac{\partial L}{\partial z_{b}^{k}}\frac{\partial z_{b}^{k}}{\partial \theta_{t}^{k}}$
\EndFor
\EndFor
\EndFor 
\end{algorithmic}}
\end{algorithm}

\subsection{VFL with XGBoost} \label{subsec:vflxgboost}
In this work, we consider training XGBoost for classification and regression tasks.
XGBoost \cite{10.1145/2939672.2939785} is a widely used boosting tree model in tabular data training because of its better interpretation, easier parameters tuning and faster training process compared with deep learning \cite{goodfellow2016deep,lecun2015deep}. Suppose that a training set with $m$ data entries consisting of the feature space $\mathcal{X} = \{x_{1}, \cdots, x_{m} \}$ and label space $\mathcal{Y} = \{ y_{1}, \cdots, y_{m} \} $. Gradients and Hessian can be calculated from Eq. \eqref{eq:grad} and Eq. \eqref{eq:hess} for each data entry, where $y_{i}^{(t-1)}$ denotes prediction of the previous tree for the $i$-th data point. 

\begin{align}
\label{eq:grad} g_{i} &= \frac{1}{1+e^{-y_{i}^{(t-1)}}} - y_{i} = \hat{y}_{i} - {y}_{i} \\
\label{eq:hess} h_{i} &= \frac{e^{-y_{i}^{(t-1)}}}{(1+e^{-y_{i}^{(t-1)}})^{2}}
\end{align}

Before training starts, the threshold value for each feature split is defined based on the predefined buckets. For building trees, the XGBoost algorithm splits each node based on whether the current depth and the number of trees have reached the predefined maximum depth and maximum tree number. If neither of the above conditions is satisfied, a new split from all splits is defined based on maximum $L_{split}$ in Eq. \eqref{eq:lsplit}, where $\lambda$ and $\gamma$ are regularization parameters.
\begin{equation}
    L_{split} = \frac{1}{2}[\frac{\sum_{i \in I_{L}}g_{i}}{\sum_{i \in I_{L}}h_{i}+\lambda} + \frac{\sum_{i \in I_{R}}g_{i}}{\sum_{i \in I_{R}}h_{i}+\lambda} - \frac{\sum_{i \in I}g_{i}}{\sum_{i \in I}h_{i}+\lambda}] - \gamma
\label{eq:lsplit} 
\end{equation}
The current node is the leaf node if the maximum $L_{split}$ is smaller than the predefined threshold. The leaf value can be calculated according to Eq. \eqref{eq:lw}.
\begin{equation}
    w = - \frac{\sum_{i \in I} g_{i}}{\sum_{i \in I} h_{i}+\lambda}.
\label{eq:lw}    
\end{equation}

\begin{figure}
\centering 
\includegraphics[width=0.46\textwidth]{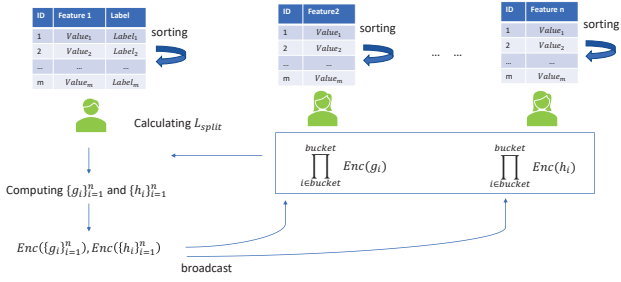}
\caption{The secureboost system.}
\label{secureboost}
\end{figure}

As shown in Figure \ref{secureboost}, Cheng \textit{et al.} \cite{cheng2021secureboost} propose a solution, where the guest client calculates the corresponding gradients and Hessians for all data points and sends them to the host clients in ciphertexts under Paillier encryption. Each host client sorts features and makes splits according to the predefined number of buckets. After that, each host client aggregates all the gradients and Hessians in one bucket and sends them to the guest client that will decrypt the aggregated values and find a maximum $L_{split}$ based on Eq.(\ref{eq:lsplit}). These steps will be done recursively until the trees reach the maximum depth or the number of trees, or when $L_{split}$ is smaller than the predefined threshold. 
Compared with \cite{cheng2021secureboost}, Tian \textit{et al.} \cite{tian2020federboost} propose a more efficient scheme in which the host clients sort their features first, and send the order of different features to the guest client with noise satisfying differential privacy. Once the guest client receives all orders of the features, all the training steps can be processed on the guest client. During the training, the labels never leave from the guest client. Besides, the guest client only knows the order of different features. Therefore, the privacy of the training set is guaranteed. 

\subsection{Additively Homomorphic encryption}

Homomorphic encryption (HE) \cite{ElGamal1985public, Paillier1999public} is a widely used encryption scheme for secure outsourced computation. An additively HE guarantees that multiple encrypted messages can be calculated without decrypting in advance.  


One of the popular schemes is the Paillier cryptosystem \cite{Paillier1999public}, which is a probabilistic encryption for public key cryptography based on the Decisional composite residuosity assumption. This work applies the Paillier cryptosystem in PIVODL for preventing the leakage of labels. 


The Paillier cryptosystem works as follows:
\begin{itemize}

\item \textbf{Key generation:} Randomly select two large prime numbers $p$ and $q$ s.t. $\mathrm{gcd}(pq, (p-1)(q-1))=1$. Let $n=pq$ and $\lambda = \mathrm{lcm}(p-1, q-1)$. After that, randomly choose an integer $g \in \mathbb{Z}_{n^{2}}^{*}$ and compute $\mu = (L(g^{\lambda}\mod n^{2}))^{-1}\mod n$, where $L$ is a function defined as $L(x)=\frac{x-1}{n}$. The public key $pk$ and secret key $sk$ are $(n, g)$ and $(\lambda, \mu)$, respectively.

\item \textbf{Encryption:} To encrypt a message $m\in \mathbb{Z}_{n}^{*}$, choose a random number $r\in \mathbb{Z}_{n}^{*}$ as an ephemeral key, the ciphertext is calculated as $c = g^{m}\cdot r^{n} \mod n^{2}$.

\item \textbf{Decryption:} The plaintext message $m$ can only be decrypted if the secret key $(\lambda, \mu)$ is available by computing
$m = L(c^{\lambda} \mod n^{2})\cdot \mu \mod n$

\end{itemize}

The Paillier satisfies the additive homomorphic property: $\text{Enc}(m1)*\text{Enc}(m2)=\text{Enc}(m1+m2)=g^{m_{1}}r_{1}^{n} \cdot g^{m_{2}}r_{2}^{n}$.

\subsection{Differential privacy}

Differential privacy (DP) \cite{dwork2006calibrating, dwork2006our} is a data privacy protection system. It can publish statistical information while keeping individual data private. If a substitution of an arbitrary single data entity does not cause statistically distinguishable changes, the algorithm used to run dataset satisfies DP. The definition of differential privacy is as follows:
\begin{myDef}
($\epsilon$ - differential privacy \cite{dwork2006calibrating}). Given a real positive number $\epsilon$ and a randomized algorithm $\mathcal{A}$: $\mathcal{D}^{n} \rightarrow \mathcal{Y}$. Algorithm $\mathcal{A}$ provides $\epsilon$ - differential privacy, if for all data sets D, $D^{'} \in \mathcal{D}^{n}$ differs on only one entity, and all  $\mathcal{S} \subseteq \mathcal{Y}$ satisfy:
\begin{equation}
    Pr[\mathcal{A}(D)\in \mathcal{S}] \leq exp(\epsilon)\cdot Pr[\mathcal{A}(D^{'})\in \mathcal{S}]
\end{equation}
\end{myDef}

To achieve $\epsilon$ - differential privacy, some mechanisms are proposed to add designed noise to queries. In this work, we apply a relaxation of $\epsilon, \delta$ - DP, called approximate DP \cite{dwork2006our}. The formal definition is as follows:
\begin{myDef}
($(\epsilon, \delta)$ - differential privacy ). Given two real positive numbers $(\epsilon, \delta)$ and a randomized algorithm $\mathcal{A}$: $\mathcal{D}^{n} \rightarrow \mathcal{Y}$. An algorithm $\mathcal{A}$ provides $(\epsilon, \delta)$ - differential privacy if it satisfies:
\begin{equation}
    Pr[\mathcal{A}(D)\in \mathcal{S}] \leq exp(\epsilon)\cdot Pr[\mathcal{A}(D^{'})\in \mathcal{S}] + \delta
\end{equation}
\end{myDef}

Gaussian mechanism \cite{abadi2016deep,9069945,geyer2017differentially} is usually applied in DP by adding Gaussian noise $\mathcal{N}\sim N(0, \Delta^{2}\sigma^{2})$ to the output of the algorithm, where $\Delta$ is $l_{2}$ - norm sensitivity of $D$ and $\sigma \geq \sqrt{2\ln(1.25/\delta)}$. According to \cite{abadi2016deep}, this noise can only achieve $(O(q\epsilon), q\delta)$-differential privacy, where $q$ is the sampling rate per lot. To make the noise small while satisfying differential privacy, we follow the definition in \cite{abadi2016deep} and set $\sigma \geq c\frac{q\sqrt{Tlog(1/\delta)}}{\epsilon}$, $c$ is a constant, and $T$ refers to the number of steps, to achieve $(O(q \epsilon \sqrt{T}), \delta)$-differential privacy.

\section{Problem formulation}

\subsection{System model}
We assume that the clients use the private set intersection \cite{kolesnikov2017practical, pinkas2014faster} to align data IDs before training starts. We consider that the complete training set with $m$ data entries consists of a feature space $\mathcal{X} = \{x_{1}, \cdots, x_{m} \}$, each containing $d$ features, and a label space $\mathcal{Y} = \{ y_{1}, \cdots, y_{m} \} $. Besides, $n$ clients possessing at least one feature $\{ X_{j}^{c} \mid j \in \{1, \cdots, d\} \}$ with all data points and part of labels $\{ y_{i}^{(c)} \mid i\in \{ 1, \cdots, m\} \}$ choose to train a model collaboratively, where $X_{j}^{(c)}$ and $y_{i}^{(c)}$ represent the $j$-th feature and the $i$-th label owned by the $c$-th client, respectively. 

Since the assumption that all labels are held by only one party is not realistic, we consider a variant of VFL called VFL over distributed labels (VFL-DL), where labels are distributed on multiple clients. The formal definition of VFL-DL is as follows:

\begin{myDef}
(VFL-DL). For any two clients $j, j^{'} \in n $, given a training set with $m$ data points, which consists of a feature space $\mathcal{X} = \{X_{1}^{c}, \cdots, X_{d}^{c} \mid j\in n \}$ and labels $\mathcal{Y} = \{ y_{1}^{c}, \cdots, y_{m}^{c} \mid j\in n\}$:
\begin{equation}
X^{c} \neq X^{c^{'}}, y^{c} \neq y^{c^{'}}, \mathcal{I}^{c} = \mathcal{I}^{c^{'}}, \forall c \neq c^{'}
\end{equation}
\end{myDef}

The frequently used notations are summarized in Table \ref{table:notation}.

\begin{table}[]
\caption{Notations summary}
\begin{tabular}{|p{30pt}|p{200pt}|}
\hline
\centering
 \textbf{Notation}&  \textbf{Description}  \\ \hline
 $\mathcal{X}$ & feature space \\ \hline
 $X_{j}^{c}$ & the $j$-th feature owned by the $c$-th client  \\ \hline
 $x_{i}$ & $i$-th data point with $d$ features \\ \hline
 $\mathcal{Y}$&label space  \\ \hline
 $y_{i}^{c}$ & the label of the $i$-th data point owned by the $c$-th client \\ \hline
 $\mathcal{I}$ & data index \\ \hline
 $g_{i}^{c}$ & the gradient of the $i$-th data point owned by the $c$-th client \\ \hline
 $G_{j, v}^{c}$ & the aggregated \textit{left} gradient sum for feature $j$  with threshold $v$ owned by the $c$-th client \\ \hline
 $h_{i}^{c}$ & the Hessian of the $i$-th data point owned by the $c$-th client \\ \hline
 $H_{j, v}^{c}$ & the aggregated \textit{left} Hessian sum for feature $j$ with threshold $v$ owned by the $c$-th client \\ \hline
 m & number of data points \\ \hline
 n & number of clients \\ \hline
 d &  number of features \\ \hline
 b &  number of buckets (feature thresholds) \\ \hline
 $\epsilon, \delta$ & parameters of differential privacy \\ \hline
 $\Delta$ & sensitivity \\ \hline
 $L_{split}$ &  similarity or impurity score \\ \hline
 $w$ & leaf value \\ \hline
 $pk$ & public key \\ \hline
 $sk$   & secret key \\ \hline
 

\end{tabular}
\label{table:notation}
\end{table}






 

\subsection{Threat model}
In this work, we mainly consider potential threats from clients and outsiders. We assume that all participating clients during training are \textit{honest-but-curious}, which means they strictly follow the designed algorithm but try to infer the labels of other clients from the received information. Besides, we also consider that outside attackers may eavesdrop communication channels and try to reveal information of the training data during the whole VFL process. The main goal of this work is to prevent the leakage of private data, and other attacks like data poisoning and backdoor attacks, which may deteriorate model performance, are not considered here.  

\subsection{Privacy concerns}
Since the data labels are distributed across multiple clients in our framework, each client becomes 'guest' client which is responsible for both aggregating $\sum g_{i}$ and $\sum h_{i}$ and updating label predictions. This would further lead to potential node split and prediction update leakage during the execution of vertical federated XGBoost protocol, even $g_{i}$ and $h_{i}$ are protected by additive homomorphic encryption.

\subsubsection{Split leakage} \label{subsec:splitleak}
In order to find the best split with the largest $L_{split}$ value for the current tree node, each client needs to ask other clients for the corresponding summation $G_{j,v}^{c}$ and $H_{j,v}^{c}$ of all possible splits. This may result in a risk of privacy leakage from two sources: one is the summation differences among all feature splits of the same node, and the other is the summation differences between parent node and child node.

A simple example is shown in Fig. \ref{fig:da}, where client $u_{1}$ wants to compute $L_{split}$ values from $Split_{1}$ to $Split_{4}$ with the help of $u_{2}$ and $u_{3}$. 
The label information about $u_{2}$ or $u_{3}$ can be easily deduced by  \emph{differential attacks}, even if the message is encrypted. For instance, $u_{1}$ requires aggregated summation $G_{1,split_{2}}=g_{1}^{2}+g_{2}^{1}$ and  $G_{1,split_{3}}=g_{1}^{2}+g_{2}^{1}+g_{3}^{3}$ to calculate $L_{split_{2}}$ and $L_{split_{3}}$, respectively. For privacy concerns, both $g_{1}^{2}$ and $g_{3}^{3}$ are encrypted before aggregation and $u_{1}$ only achieves aggregated summations and has no knowledge of these two gradient values. However, $g_{3}$ can still be derived based on the \emph{difference} (subtract here) between $G_{1,split_{3}}$ and  $G_{1,split_{2}}$. Similarly, $g_{1}^{2}$ can be easily obtained by  $G_{1,split_{2}} - g_{2}^{1}$.


\begin{figure}
\centering 
\includegraphics[width=0.46\textwidth]{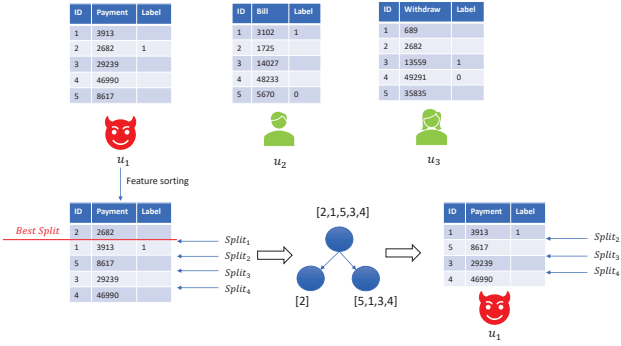} 
\caption{A simple example of split leakage} 
\label{fig:da}
\end{figure}
 
After finding the maximum $L_{split}$ for the current node, we can do the split by allocating data 2 to the left child node and the rest of data instances to the right child node. If $u_{1}$ wants to continue splitting on the right child node, it still requires to know the aggregated summation for every possible split. Specifically, differential attacks can be performed upon the $Split_{3}$ between the parent node ($g_{1} + g_{5} + g_{2}$) and right child node ($g_{1} + g_{5}$ ), thus, $g_{2}$ can be easily deduced by $(g_{1} + g_{5} + g_{2}) - (g_{1} + g_{5})$.


\subsubsection{Prediction update leakage} \label{subsec:predictionleak}
Except for split leakages, private label information may also be disclosed in updating the predictions. As introduced in Section \ref{subsec:vflxgboost}, XGBoost is an ensemble boosting algorithm that requires to sequentially  construct multiple decision trees for regression or classification tasks. And the label prediction $\widehat{y}$ for each data sample should be updated based on the leaf value of the current tree, whenever a new decision tree is being built. Unlike in conventional VFL where the guest client performs all leaf weights calculations and label prediction updates; in PIVODL, labels are stored on multiple clients based on their split features of the parent nodes. As a result, it is very likely that some clients will receive predicted labels of other clients.

An example is shown in Fig. \ref{fig:label_update}, where three leaf nodes are stored on three clients $u_{1}$, $u_{2}$ and $u_{3}$, respectively. It is clear to see that $u_{1}$ knows the leaf weight of data 1, 2, 6 and 9, thus, the label prediction of the current decision tree for data 1 on $u_{2}$ and data 6 on $u_{3}$ are leaked to $u_{1}$. There is no doubt that the predicted values will be closer to the true labels with a high probability over the training process. Therefore, a more accurate guess can be made for specific data samples along with the model convergence, if the clients can get some of their leaf weights.

\begin{figure}
\centering 
\includegraphics[width=0.46\textwidth]{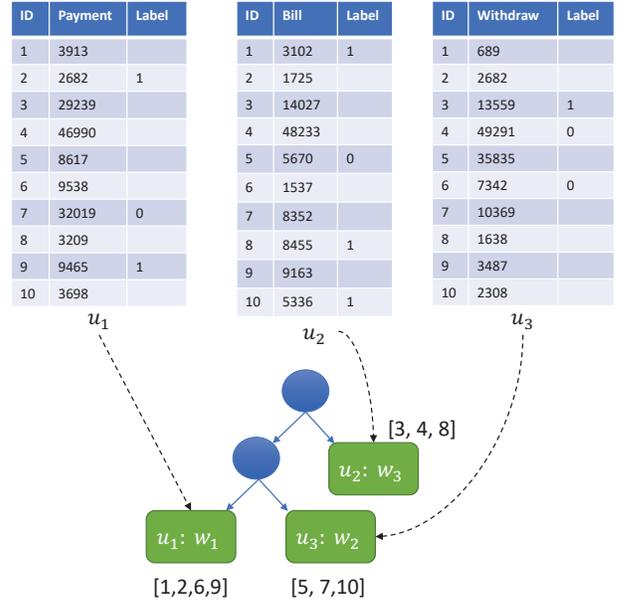} 
\caption{Label prediction update in VFL with the true data labels distributed across multiple clients.} 
\label{fig:label_update}
\end{figure}

\section{The proposed system}


In order to alleviate the aforementioned two leakage issues, we propose a novel privacy-preserving vertical federated XGBoost learning system over distributed labels (PIVODL), in which an ensemble of boosting trees can be constructed without revealing intermediate gradient and Hessian information. In addition, a partial differential privacy scheme is proposed to protect the predicted labels with an acceptable drop of the model performance. Moreover, we also impose attack inferences to our system and verify that no client inside the system is is able to receive or deduce label information from other clients.

\subsection{A secure node split protocol}
As discussed in the previous section, private label information may be disclosed during the calculation of $L_{split}$ and the branch splits. This is because the client used for \emph{node split} also \emph{acquire} all the data ID information for possible feature splits of current tree node. Therefore, as long as the client knows the intermediate aggregation results of $G_{j, v}^{c}$ and $H_{j, v}^{c}$ from other clients, $g_{i}^{c}$ and $h_{i}^{c}$ can be easily deduced through the differences between two adjacent summation results, regardless of whether any cryptographic operations are used.

In order to prevent label information from being leaked during node splits, the proposed PIVODL framework performs feature splits and aggregation summations on two \emph{separate} clients. One is called the \emph{source client} and the other is called the \emph{split client}. In the following, we define the roles of the source and split clients. 

\begin{myDef}
The \textbf{source client} is the split owner that contains all data IDs of the local features. However, it does not have any information of the corresponding summation values $G_{j, v}$ and $H_{j,v}$, and thus, is not able to calculate the similarity scores in Eq. \eqref{eq:lsplit} and leaf weights in Eq. \eqref{eq:lw}.
\label{source_clients}
\end{myDef}

\begin{myDef}
The \textbf{split client}, on the other hand, is used to compute both the similarity scores in Eq. \eqref{eq:lsplit} and leaf weights in Eq. \eqref{eq:lw}, however, it has no idea of their corresponding data IDs on other clients.
\label{split_clients}
\end{myDef}

By making use of the source and split clients, the tree nodes can be split without leaking the data privacy in the following steps:
\begin{enumerate}
    \item Each client locally generates its unique \textit{Paillier} key pairs and exchanges the public key with each other. Meanwhile, the data IDs of the samples with labels $I^{c}$ are also exchanged with each other for privacy concern, which will be discussed in the next section.
    
    \item Each client $c$ computes locally $G^{c}=\sum_{i}g_{i}^{c}$ and $H^{c}=\sum_{i}h_{i}^{c}$ for the current tree node according the to available data labels $y_{i}^{c}$. Randomly select two clients to be the \emph{aggregation client} and \emph{encryption client}. Each client uses the public key of the \emph{encryption client} to encrypt $G^{c}$ and $H^{c}$ and sends the encrypted $\langle G^{c} \rangle$ and $\langle H^{c} \rangle$ to the \emph{aggregation client} for summation: $\langle G \rangle=\prod_{c}\langle G^{c} \rangle= \langle \sum_{c}G^{c} \rangle$, $\langle H \rangle=\prod_{c}\langle H^{c} \rangle= \langle \sum_{c}H^{c} \rangle$. After that, the \textit{aggregation client} sends $\langle G \rangle$ and $\langle H \rangle$ to the \emph{encryption client} for decryption: $G=Dec\langle G \rangle$,  $H=Dec\langle H \rangle$. Finally, the \emph{encryption client} can broadcast the decrypted $G$ and $H$ for calculating the information gain of the current node by $\frac{1}{2}\frac{G}{H+\lambda}$.
    
    \item \label{setsourceclient} All client are regarded as the source clients and sort their data samples according to the split buckets (feature thresholds) of all the local features. In addition, each source client randomly selects one client from other clients to be its \emph{unique} split client. And then, each source client $c$ computes local left branch sum $G_{j,v}^{c}=\sum_{i\in \{ i | x_{i,j} < v \}}g_{i}^{c}$ and $H_{j,v}^{c}=\sum_{i\in \{ i | x_{i,j} < v \}}h_{i}^{c}$ based on its owned labels $y_{i,j}^{c}$, and sends the \emph{intersected} missing data IDs $I_{j,v}^{{c}'}=I_{j,v}^{c} \bigcap I^{{c}'},{c}' \in [1,m],{c}'\neq c$ together with its split client ID and temporary record number $R_{j, v}$ to other clients.
    
    \item After receiving intersected data IDs and corresponding split client ID from source clients, each client $c$ calculates $G_{j,v}^{c}=\sum_{i\in \{ i | i \in I_{j,v}^{c} \}}g_{i}^{c}$ and $H_{j,v}^{c}=\sum_{i\in \{ i | i \in I_{j,v}^{c} \}}h_{i}^{c}$, and use the public key of its corresponding split client to encrypt the summation values. And then each client $c$ returns the encrypted $\langle G_{j,v}^{c} \rangle$ and $\langle H_{j,v}^{c} \rangle$ back to the corresponding source client.
    
    \item After the source client receives all $\langle G_{j,v}^{c} \rangle$ and $\langle H_{j,v}^{c} \rangle$ from other clients, it will aggregate (sum) all received encrypted summation with the local gradient and Hessian summation to get the final encrypted $\langle G_{j,v} \rangle=\prod_{c}\langle G_{j,v}^{c} \rangle= \langle \sum_{c}G_{j,v}^{c} \rangle$ and $\langle H_{j,v} \rangle=\prod_{c}\langle H_{j,v}^{c} \rangle= \langle \sum_{c}H_{j,v}^{c} \rangle$ for all bucket splits.
    
    \item Each source client sends $\langle G_{j,v} \rangle$ and $\langle H_{j,v} \rangle$ to its corresponding split client for decryption and the split clients can compute the impurity scores by $L_{split_{j,v}}=\frac{1}{2}[\frac{G_{j,v}}{H_{j,v}+\lambda}+\frac{G-G_{j,v}}{H-H_{j,v}+\lambda}-\frac{G}{H+\lambda}]-\gamma$. All the selected split clients get the largest local impurity score at first and then broadcast these values for further comparison to achieve the largest global impurity score $L_{split_{j,v}^{max}}$. The split client that owns $L_{split_{j,v}^{max}}$ would send the temporary record number $R_{j,v}$ back to the source client. The source can get the best split feature $j$ and threshold $v$ according to received $R_{j,v}$ and privately build a lookup table \cite{cheng2021secureboost} to record $j$ and $v$ with a unique \emph{record ID}. And then the tree node can be constructed by the source client ID and the record ID as shown in Fig. \ref{fig:node_split}.
    
    \item If the current tree node is not the leaf node, we will continue splitting either the left or right child node. If the left node becomes the current node for splitting, the split client sets $G=G_{j,v}$, $H=H_{j,v}$, and sends them to all other clients. Otherwise, the split client sets $G=G-G_{j,v}$, $H=H-H_{j,v}$, and sends them to all other clients. In addition, the source client sends the left split node data IDs $I_{j,v}^{L}$ or the right split node data IDs $I_{j,v}^{R}$ to other clients. Then, each client computes the current information gain $\frac{1}{2}\frac{G}{H+\lambda}$. Go to step \eqref{setsourceclient} and repeat the process.
\end{enumerate}

\begin{figure}
\centering 
\includegraphics[width=0.46\textwidth]{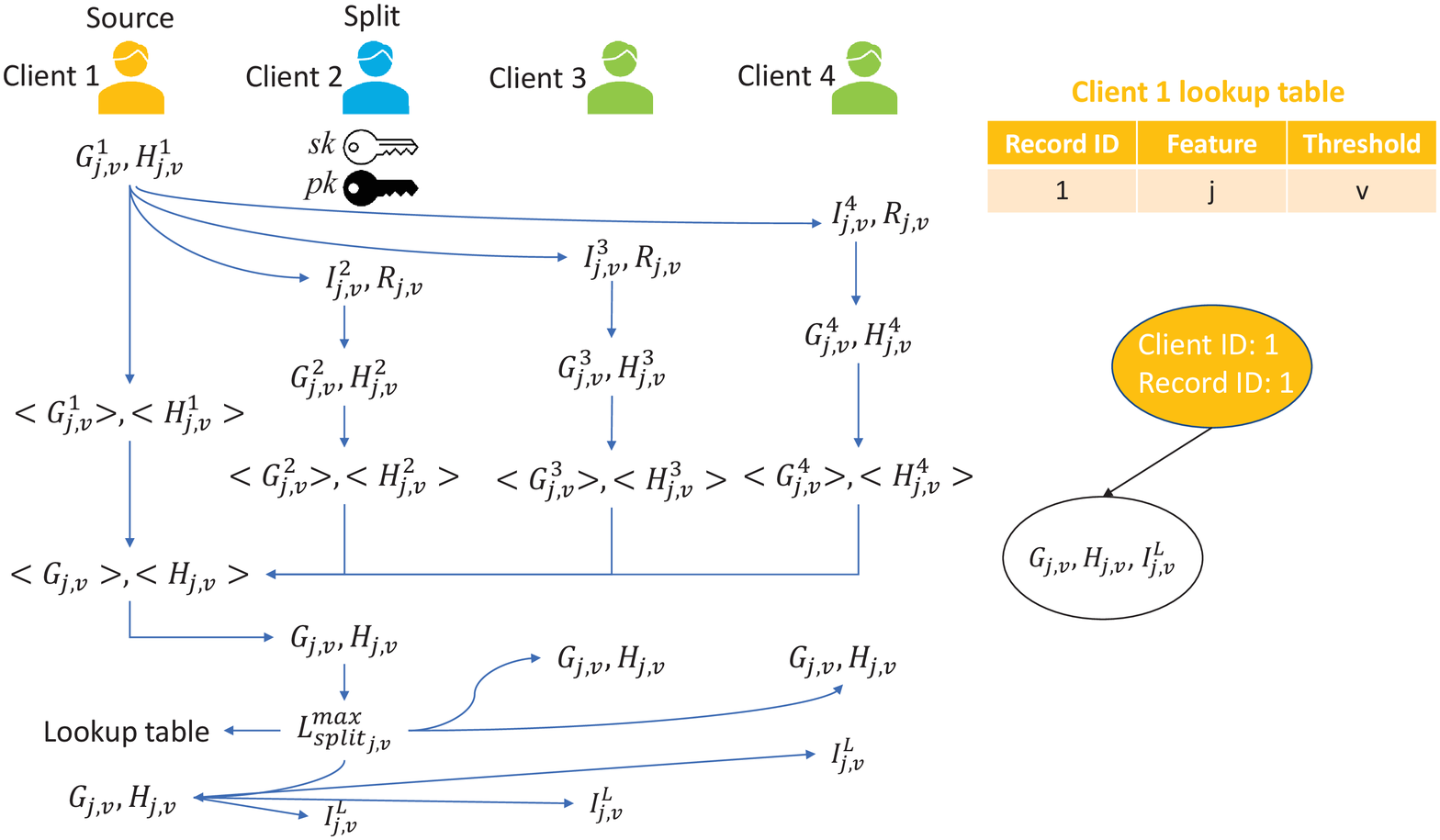} 
\caption{Secure node split for the left branch, where we assume the feature $j$ and threshold $v$ will give the best split with the largest impurity score $L_{split}$.} 
\label{fig:node_split}
\end{figure}

It is clear to see that above proposed secure node split protocol encrypts the intermediate summation values $G_{j,v}$ and $H_{j,v}$ for privacy concern and is able to effectively defend differential attacks. The source client contains all the split node IDs and their corresponding split features and thresholds. However, as shown in Fig. \ref{fig:node_split}, the source client $1$ can only get the encrypted sum $\langle G_{j,v} \rangle$ and $\langle H_{j,v} \rangle$ for feature $j$ and threshold $v$. Thus, differential attack does not work in this scenarios since only the split client $2$ has the secret key to get the plaintext $G$ and $H$. On the other hand, the split client $2$ only knows local intersected data IDs $I^{2}_{j,v}$ and the final summation values $G$ and $H$, but it has no idea of the corresponding data IDs. In addition, the split client $2$ only receive the temporary record number $R_{j,v}$ (this is just a arbitrary number defined by the source client, 'temporary' means this number would be reset for the next node split) of the split and has no knowledge of the corresponding real split feature and threshold on the source client. 

After all split clients decrypt $G_{j,v}$ and $H_{j,v}$ for all available features $j$ and threshold $v$, they can calculate $L_{split}$ according to Eq. \eqref{eq:lsplit}. In order to reduce  both potential privacy leakage and communication costs, each split client compares and gets the largest $L_{split}$ value locally at first. The advantage of doing this is that each split client just needs to broadcast only one impurity score $L_{split}$ for comparison. If the current tree node continues splitting to the left branch as shown in Fig. \ref{fig:node_split}, the split client needs to transmit the decrypted $G_{j,v}$ and $H_{j,v}$ to all other clients as $G$ and $H$ of the next tree node. In addition, the source client needs to send the data IDs $I^{L}_{j,v}$ of the left split to other clients. Therefore, each client can deduce the gradient and Hessian values from other clients through the received $G_{j,v}$, $H_{j,v}$ and $I^{L}_{j,v}$ as discussed in Section \ref{subsec:splitleak}.
In order to alleviate this privacy leakage issue, we adopt a simple data instance threshold strategy. For example, if the number of received intersected data IDs is less than the pre-defined instance threshold, this client will reject to compute the local $G_{j,v}^{c}$ and this split will be removed during the node split. It should be mentioned that continue splitting is robust to differential attacks, since only one $G_{j,v}$ and $H_{j,v}$ are broadcast for each split.
The overall secure node split protocol is also shown in Algorithm \ref{alg:securesplit}.

\begin{algorithm}[htbp]\footnotesize{
\caption{Secure node split protocol } 
\algblock{Begin}{End}
\label{alg:securesplit}
\begin{algorithmic}[1]
\State{\textbf{Input: }$I$, data IDs of current tree node}
\State{\textbf{Input: }$G$, gradient sum of current tree node}
\State{\textbf{Input: }$H$, Hessian sum of current tree node}
\State{\textbf{Input: }$T$, instance threshold of current tree node} \\

\State{\textbf{\textit{Split client sets}} ${C}'= \emptyset $}
\For{each \textbf{\textit{source client}} $c=1,2,...m$}
\State{Randomly select a \textbf{\textit{split client ${c}'$}}, ${c}' \in [1,m],{c}'\neq c $}
\State{$L_{split}^{{c}'} \leftarrow 0 $, ${C}' \leftarrow {C}' \cup {c}' $}
\EndFor \\

\For{each \textbf{\textit{source client}} $c=1,2,...m$}
\State{$R_{j,v}^{c} \leftarrow 0$ }
\For{each feature $j=1,2,...d^{c}$}
\For{each threshold $v=1,2,...b_{j}$}
\State{$pk \leftarrow pk \textnormal{ of client } {c}'$}
\State{$G_{j,v}^{c} \leftarrow  \sum_{i\in \{ i | x_{i,j} < v \}}g_{i}^{c}$, $\langle G_{j,v}^{c} \rangle  \leftarrow  Enc_{pk}(G_{j,v}^{c})$}
\State{$H_{j,v}^{c} \leftarrow  \sum_{i\in \{ i | x_{i,j} < v \}}h_{i}^{c}$, $\langle H_{j,v}^{c} \rangle  \leftarrow  Enc_{pk}(H_{j,v}^{c})$}
\State{$R_{j,v}^{c} \leftarrow R_{j,v}^{c} + 1$}
\For{each \textbf{\textit{client}} ${c}''=1,2,...m,{c}''\neq c$}
\State{$I_{j,v}^{{c}''} \leftarrow I_{j,v}^{c} \bigcap I^{{c}''}$}
\State{Send $I_{j,v}^{{c}''}$, $R_{j,v}^{c}$ to client ${c}''$: }
\State{$pk \leftarrow pk \textnormal{ of client } {c}'$}
\If{$|I_{j,v}^{{c}''}| \geq T$}
\State{$G_{j,v}^{{c}''} \leftarrow  \sum_{i\in I_{j,v}^{{c}''}}g_{i}^{{c}''}$, $\langle G_{j,v}^{{c}''} \rangle  \leftarrow  Enc_{pk}(G_{j,v}^{{c}''})$}
\State{$H_{j,v}^{{c}''} \leftarrow  \sum_{i\in I_{j,v}^{{c}''}}h_{i}^{{c}''}$, $\langle H_{j,v}^{{c}''} \rangle  \leftarrow  Enc_{pk}(H_{j,v}^{{c}''})$}
\State{Return $\langle G_{j,v}^{{c}''} \rangle$ and $\langle H_{j,v}^{{c}''} \rangle$ to \textbf{\textit{source client $c$}}}
\EndIf
\EndFor
\If{\textbf{\textit{Source client $c$}} receives $m-1$ $\langle G_{j,v}^{{c}''} \rangle$ and $\langle H_{j,v}^{{c}''} \rangle$}
\State{$\langle G_{j,v} \rangle=\prod_{c}\langle G_{j,v}^{c} \rangle=\langle \sum_{c} G_{j,v}^{c} \rangle$}
\State{$\langle H_{j,v} \rangle=\prod_{c}\langle H_{j,v}^{c} \rangle=\langle \sum_{c} H_{j,v}^{c} \rangle$}
\State{Send $\langle G_{j,v} \rangle$ and $\langle H_{j,v} \rangle$ to \textbf{\textit{split client ${c}'$}}:}
\State{$G_{j,v} \leftarrow Dec \langle G_{j,v} \rangle$}
\State{$H_{j,v} \leftarrow Dec \langle H_{j,v} \rangle$}
\State{$L_{split_{j,v}}=\frac{1}{2}[\frac{G_{j,v}}{H_{j,v}+\lambda}+\frac{G-G_{j,v}}{H-H_{j,v}+\lambda}-\frac{G}{H+\lambda}]-\gamma$}
\If{$L_{split_{j,v}} > L_{split}^{{c}'}$}
\State{$L_{split}^{{c}'} \leftarrow L_{split_{j,v}}$, $R_{j,v}^{c,best} \leftarrow R_{j,v}^{c}$}
\EndIf
\EndIf
\EndFor
\EndFor
\EndFor \\

\State{\textbf{Output: } the \textbf{\textit{split client ${c}'$}} with the largest impurity score $L_{split}^{{c}'}$}

\end{algorithmic}}
\end{algorithm}

\subsection{Construction of private tree nodes}
After the split client ${c}'$ with the largest impurity score is determined, ${c}'$ will send the temporary record number $R_{j,v}^{c, best}$ of this split to the corresponding source client $c$. And then the source client $c$ can add the feature $j$ and threshold $v$ with a unique record ID into the local lookup table \cite{cheng2021secureboost}, as shown in Fig. \ref{fig:node_split}. After that, $c$ broadcasts the unique record ID to any other client that can annotate the split of current tree node with client $c$'s ID and the record ID.

Meanwhile, the source client $c$ can check and tell the split client ${c}'$ whether the split child nodes are leaf or not based on the condition shown in line 5 and line 15 of Algorithm \ref{alg:checkleaf}, respectively. If the child node is not leaf, the source client $c$ broadcasts the split data IDs and the split client ${c}'$ broadcasts the corresponding summation of gradients and Hessian for continuing splitting the child node. If the child node is determined to be a leaf node, the split client ${c}'$ computes the local leaf weight and stores it in a \emph{received} lookup table with source client $c$ ID and record ID.

\begin{algorithm}[htbp]\footnotesize{
\caption{Construction of private tree nodes, where ${c}'$ is the split client with the largest impurity score, $c$ is the source client of ${c}'$, $R_{j,v}^{c,best}$ is the temporary record number of the best split on $c$, and $T_{sample}$ is the minimum data samples for each decision node.} 
\algblock{Begin}{End}
\label{alg:checkleaf}
\begin{algorithmic}[1]

\State{${c}'$ sends $R_{j,v}^{c,best}$ and split branch to $c$}
\State{$c$ adds feature $j$ and threshold $v$ with a unique record ID into the local lookup table}
\State{$c$ broadcasts the record ID and each client can annotate the split of current tree node with client $c$ ID and record ID}
\If{Split branch is left}
\If{Reach the maximum depth or $|I_{j,v}| < T_{sample}$}
\State{$c$ tells ${c}'$ left split node is the leaf node and stops splitting}
\State{$w_{j,v}^{L}=-\frac{G_{j,v}}{H_{j,v}+\lambda}$}
\State{${c}'$ records left leaf value $w_{j,v}^{L}$ into \emph{received} lookup table with record ID and client $c$ ID}
\Else
\State{Continue splitting left child node:}
\State{$c$ broadcasts $I_{j,v}$ to all other clients}
\State{${c}'$ broadcasts $G_{j,v}$ and $H_{j,v}$ to all other clients}
\EndIf
\ElsIf{Split branch is right}
\If{Reach the maximum depth or $|I-I_{j,v}| < T_{sample}$}
\State{$c$ tells ${c}'$ right split node is the leaf node and stops splitting}
\State{$w_{j,v}^{R}=-\frac{G-G_{j,v}}{H-H_{j,v}+\lambda}$}
\State{${c}'$ records right leaf value $w_{j,v}^{R}$ into \emph{received} lookup table with record ID and client $c$ ID}
\Else
\State{Continue splitting right child node:}
\State{$c$ broadcasts $I-I_{j,v}$ to all other clients}
\State{${c}'$ broadcasts $G-G_{j,v}$ and $H-H_{j,v}$ to all other clients}
\EndIf
\EndIf
\end{algorithmic}}
\end{algorithm}

As a result, no client except the source client knows the feature $j$ and threshold $v$ of the tree node split, since the lookup table is stored locally without sharing with any other clients. And when it turns to new instance prediction, the source client needs to combine with the split client to get the leaf value. For instance, as shown a simple decision tree in Fig. \ref{fig:lookup}, client $1$ is the source client of the root client and client $2$ is its corresponding split client. If a new data sample with age 30 and weight 150 comes for prediction, it will be sent to client $1$ for judgement. And then client $1$ uses its own lookup table to determine that this instance should go to the left branch ($30 < 35$). After that, this data is sent to client $2$ to get its label prediction 0.41 ($150 > 120$, go to right leaf). Also considering potential label leakage during data prediction for ensemble decision trees, we adopt secure aggregation scheme proposed in \cite{10.1145/3133956.3133982} for secure multi-party computation. Therefore, the final prediction for the data instance can be derived without revealing any leaf values. 

\begin{figure}
\centering 
\includegraphics[width=0.46\textwidth]{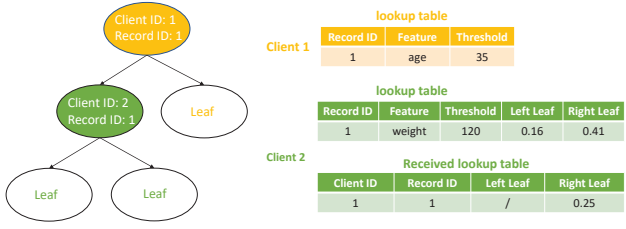} 
\caption{Each split node is recorded with the source client ID and a unique record ID. Both the lookup table and received lookup table are stored locally without sharing with any other client. The received lookup table needs to contain the source client ID and the corresponding record ID.} 
\label{fig:lookup}
\end{figure}



\subsection{Partial differential privacy}
Suppose that the child node is a leaf node because training reaches either the maximum depth or the minimum node data samples. Then the leaf weight used for prediction will be calculated by split client (shown in line 5 and line 14 in Algorithm \ref{alg:checkleaf}). There is no doubt that as the training proceeds, the \emph{accumulated} label prediction $\widehat{y_{i}}$ will be closer to the true data label $y_{i}$, and since the source client knows all the split data IDs, as long as it receives the leaf weight for prediction update, 'parts of' \footnote{Here we call it 'parts of', because XGboost is an ensemble tree algorithm. A client can get the real label predictions only if the leaf weights of all boosting trees of data samples are revealed to the client.} the true labels of the corresponding data samples may have the risk of being revealed. However, all participating clients require the leaf value $w$ to update the corresponding local predictions. Thus, it is inevitable that $w$ will be sent to the source client, which means the source client has a potential opportunity to guess the true labels through the received $w$.



In order to deal with this prediction update leakage issue, we propose a partial differential privacy (DP) mechanism, where the leaf value is perturbed with Gaussian noise $\mathcal{N}$ before being sent to the source client. Since the calculation of the leaf value $w$ is based on the summation of $g$ and $h$ (the number of $g$ and $h$ is not fixed), its range is unknown beforehand. Before adding noise to the leaf values, we clip the leaf values and set the clip value as the sensitivity of differential privacy. The sensitivity of $w$ is calculated according to Eq. \eqref{eq:differential_privacy}, where $C$ refers to the predefined clip value and $\{ G,H\},\{G^{'},H^{'} \}$ are two sets, differing in one pair of $g$ and $h$.

\begin{equation}
\begin{aligned}
\Delta_{w} &=  \mathop{max}\limits_{\{ G,H\},\{G^{'},H^{'} \}} \| w^{\{ G,H\}} - w^{\{G^{'},H^{'}\}} \| \\ &=2C, \; \forall  w^{\{ G,H\}},w^{\{G^{'},H^{'}\}} \leq C \\  
\end{aligned}
\label{eq:differential_privacy}
\end{equation}

By \textit{partial} DP, we mean here that DP is apply on $w$ values that are sent to the source client only, since other clients only have the intersected local data IDs and have no idea about any other data IDs. The advantage of using this strategy is that the correct distributed label predictions can be achieved to the greatest extent, while preventing the source client from guessing the true labels.
The details of the partial DP mechanism is presented in Algorithm \ref{alg:partialdp}:

\begin{algorithm}[htbp]\footnotesize{
\caption{Partial differential privacy. ${c}'$ is the split client and $c$ is the corresponding source client} 
\algblock{Begin}{End}
\label{alg:partialdp}
\begin{algorithmic}[1]
\State{\textbf{Input: }$w_{j,v}$, computed leaf weight on split client ${c}'$} \\
\For{each ${c}''=1,2,...m, {c}'' \neq {c}'$}
\If{${c}'=c$}
\State{$\hat{w_{j,v}} \leftarrow w_{j,v} / max(1, \frac{||w_{j,v}||}{C})$}
\State{$\hat{w_{j,v}} \leftarrow \hat{w_{j,v}} + \mathcal{N}\sim N(0, {2C}^{2}\sigma^{2})$}
\State{Send perturbed $\hat{w_{j,v}}$ to $c$}
\Else
\State{Send $w_{j,v}$ to ${c}''$}
\EndIf
\EndFor
\end{algorithmic}}
\end{algorithm}




\subsection{Security analysis}

Our proposed PIVODL framework avoids revealing both data features and labels on each participating client, and in this section, we will make a detailed security analysis of the leakage sources and protection strategies.

During the training process, each client sorts the data instances based on the features and their corresponding bucket thresholds which are stored privately on the local devices. In order to find the maximum impurity score $L_{split}$, the source clients require the corresponding $g$ and $h$ from other clients, since labels are distributed on multiple clients. The true data labels can be deduced from the computed gradients $g$.
\begin{theorem}
Given a gradient $g$, the true label of corresponding data point could be inferred.
\end{theorem}
\begin{proof}
According to Eq. \eqref{eq:grad}, adversarial attackers are able to infer the true labels of the $i$-th data sample by $y_{i}=\hat{y_{i}}-g_{i}$ if $g_{i}$ is known.
\end{proof}

To tackle this privacy issue, Paillier encryption is adopted in this work to encrypt all $\sum g$ and $\sum h$ for each bucket before sending them for summation aggregation. However, it is possible that the specific gradient information $g_{i}$ can still be inferred by a differential attack, if the clients have data ID information of its adjacent bucket splits. 
\begin{theorem}
There is still a high risk that the gradient $g_{i}$ of the $i$-th data sample may be deduced with a differential attack, even if the intermediate $\sum g$ and $\sum h$ are encrypted.
\end{theorem}
\begin{proof}

Assume $L_{n}$ and $L_{{n}'}$ are two adjacent data sample sets for two possible splits that only differs in one data sample $i$ (or a few data samples). Although the client only knows the aggregated summation $\sum_{i\in L_{n}}g_{i}$ and $\sum_{i\in L_{{n}'}}g_{i}$, $g_{i}$ can be easily derived with a differential attack $g_{i}=\sum_{i\in L_{n}}g_{i} - \sum_{i\in L_{{n}'}}g_{i}$.
\end{proof}

Therefore, it is critical to ensure that the source client and split client perform the bucket split and aggregation separately. And then the source client only knows the data IDs of each feature split but has no knowledge of their corresponding summation values, the split client holds the summation values but does not know their data IDs.


Before building the next boosting tree, the predictions of all data samples need to be updated based on the corresponding leaf weights. It will not be surprising that the predicted labels are the same as the true labels as the training proceeds. Therefore, attackers can guess the correct labels with a high confidence level through the received leaf weights. 
\begin{theorem}
Given multiple boosting trees' leaf values, it is possible to deduce the true labels of the corresponding data points.
\end{theorem}
\begin{proof}
The label predictions $\hat{y}$ in a classification problem can be calculated  with equation $\hat{y}= \sigma(\hat{y_{0}}+\eta\cdot W_{1} + \cdots + \eta\cdot W_{t})$, where $\eta$ refers to the learning rate, $W_{t}$ is the leaf weight of $t$-th decision tree, and $\sigma$ is the activation function. Therefore, once the source client knows parts of or even all the leaf values $W$, the true labels have a high risk of being inferred.
\end{proof}

In order to address this issue, the split client applies the above-mentioned partial DP mechanism before sending the leaf values to the source client. In this way, we can reduce the risk for the source client to correctly guess the true labels from other clients.



\section{Experimental results}
In this section, we first provide the experimental settings and then present the experimental results, followed by a discussion of the learning performance. After that, we evaluate the training time when implementing our privacy-preservation method. At the end, the communication cost and label prediction inference of PIVODL would be described.

\subsection{Experimental Settings}
Three common public datasest are used in the experimental studies, where the first two are classification tasks and the third is a regression task.

\textbf{Credit card} \cite{yeh2009comparisons}: It is a credit scoring dataset that aims to predict if a person will make payment on time. It contains in total 30,000 data samples with 23 features. 

\textbf{Bank marketing} \cite{moro2014data}: The data is related to direct marketing campaigns of a Portuguese banking institution. The prediction goal is to evaluate whether the client will subscribe a term deposit. It consists of 45211 instances and 17 features.

\textbf{Appliances energy prediction} \cite{candanedo2017data}: It is a regression dataset of energy consumption in a low energy building, which has 19735 data instances and 29 attributes. 

We partition each dataset into training data and test data. The training and test dataset occupies 80\% and 20\% of the entire data samples, respectively. Other experimental settings like the number of clients, the number of the maximum depth of the boosting tree are presented in Table \ref{t:settings}. Besides, each experiment is repeated for five times independently and the results of the mean values together with their standard deviations are plotted in the corresponding figures.

\begin{table}[]
\centering
\caption{Experimental settings of the PIVODL system}
\begin{tabular}{c|c|c}
\toprule
Hyperparameters   & Range            & Default \\
\toprule
Number of clients & {[}2,4,6,8,10{]} & 4       \\
Maximum depth     & {[}2,3,4,5,6{]}  & 3       \\
Number of trees   & {[}2,3,4,5,6{]}  & 5       \\
Learning rate     & /                & 0.3     \\
Regularization    & /                & 1       \\
Number of buckets & /                & 32      \\
Encrypt key size  & /                & 512     \\
$\epsilon$           & {[}2,4,6,8,10{]} & 8       \\
Sensitivity clip  & /                & 2       \\
Sample threshold  & /                & 10    \\
\bottomrule
\end{tabular}
\label{t:settings}
\end{table}

\subsection{Sensitivity analysis}
Here, we empirically analyze the change of the learning performances as the number of participating clients, the maximum depth of the tree structure, the number of trees, and $\epsilon$ in the partial DP change. The results in terms of the mean, the best, and the worst performance over five independent runs are plotted in Fig. \ref{fig:performclients}, \ref{fig:performdepth}, \ref{fig:performtrees}, and Fig. \ref{fig:performanceeps}, respectively. 

From the results, we can see that varying the number of clients does not affect the training performance very much, especially in the experiments without applying the partial DP. As shown in Fig. \ref{fig:performclients}, it is clear to see that the test performances on these three datasets are relatively insensitive to different numbers of clients. These results make sense since non-parametric models trained in the VFL setting have nearly the same performance as those trained in the standard centralized learning \cite{zhu2021federated}. By contrast, however, the average test accuracy on the Credit dataset drops from 0.82 to 0.8 when the number of clients is reduced to 2. And it is surprising to see that the root mean squared error (RMSE) on the test data of the Energy dataset with DP is slightly higher than those without DP over different number of clients. The reason for this is that the model may have overfit the training data without applying DP. For instance, when the number of clients is 4 (the default value), the training RMSE with DP is about 70.9, which is higher than that without DP (66.6). By contrast, the test RMSE with DP is approximately 0.8\% lower than that without DP. Overall, the performances with DP and without DP are almost the same and the test performance degradation resulting from the DP is negligible. Besides, the performance of the proposed PIVODL algorithm is rather stable in different runs.

\begin{figure}[!t]
\begin{minipage}[t]{1\linewidth}
\centering
\subfigure[Credit card and Bank marketing]{
\begin{minipage}[b]{0.46\textwidth}
\includegraphics[width=1\textwidth]{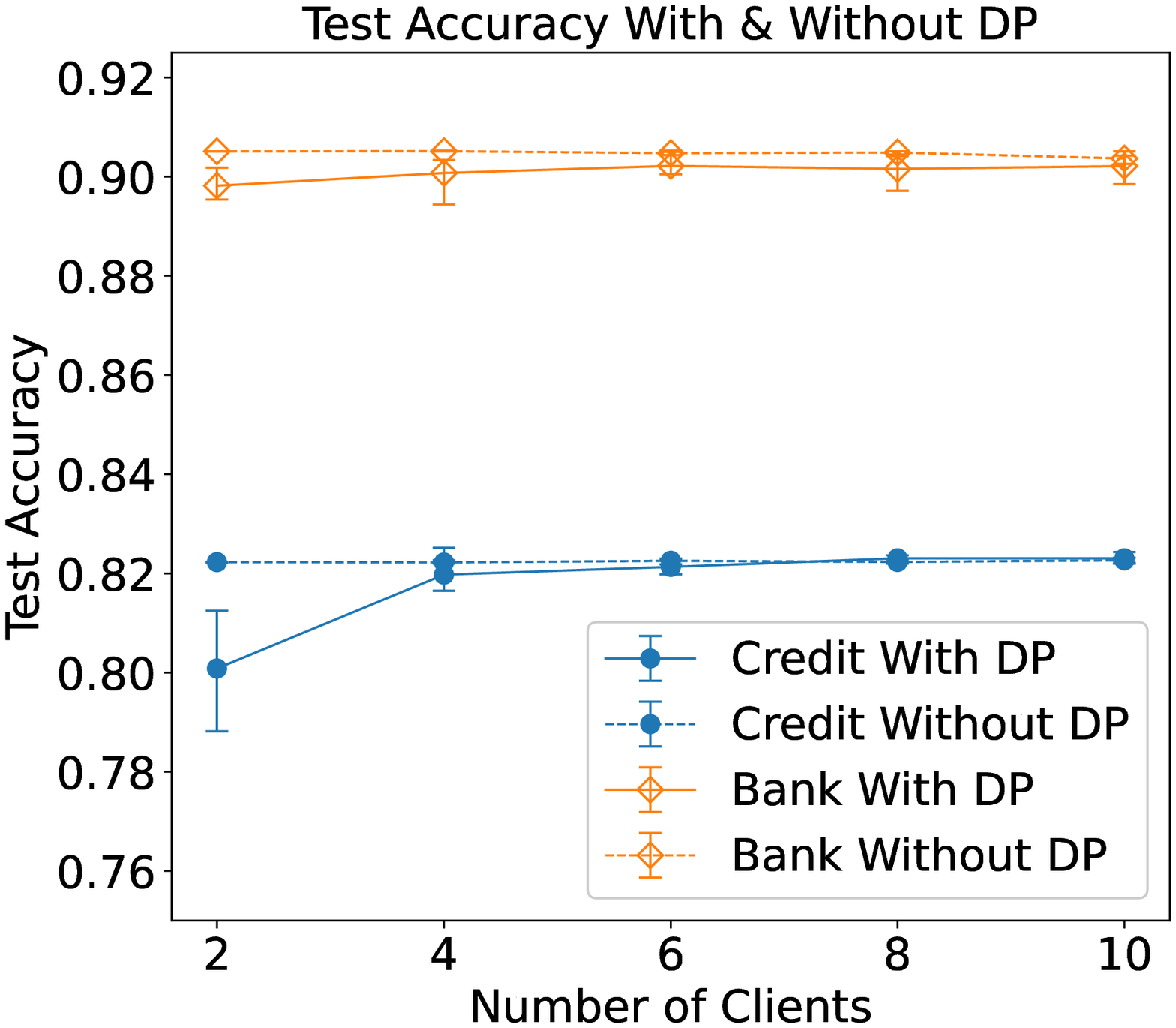}
\end{minipage}
}
\centering
\subfigure[Appliances energy prediction]{
\begin{minipage}[b]{0.46\textwidth}
\includegraphics[width=1\textwidth]{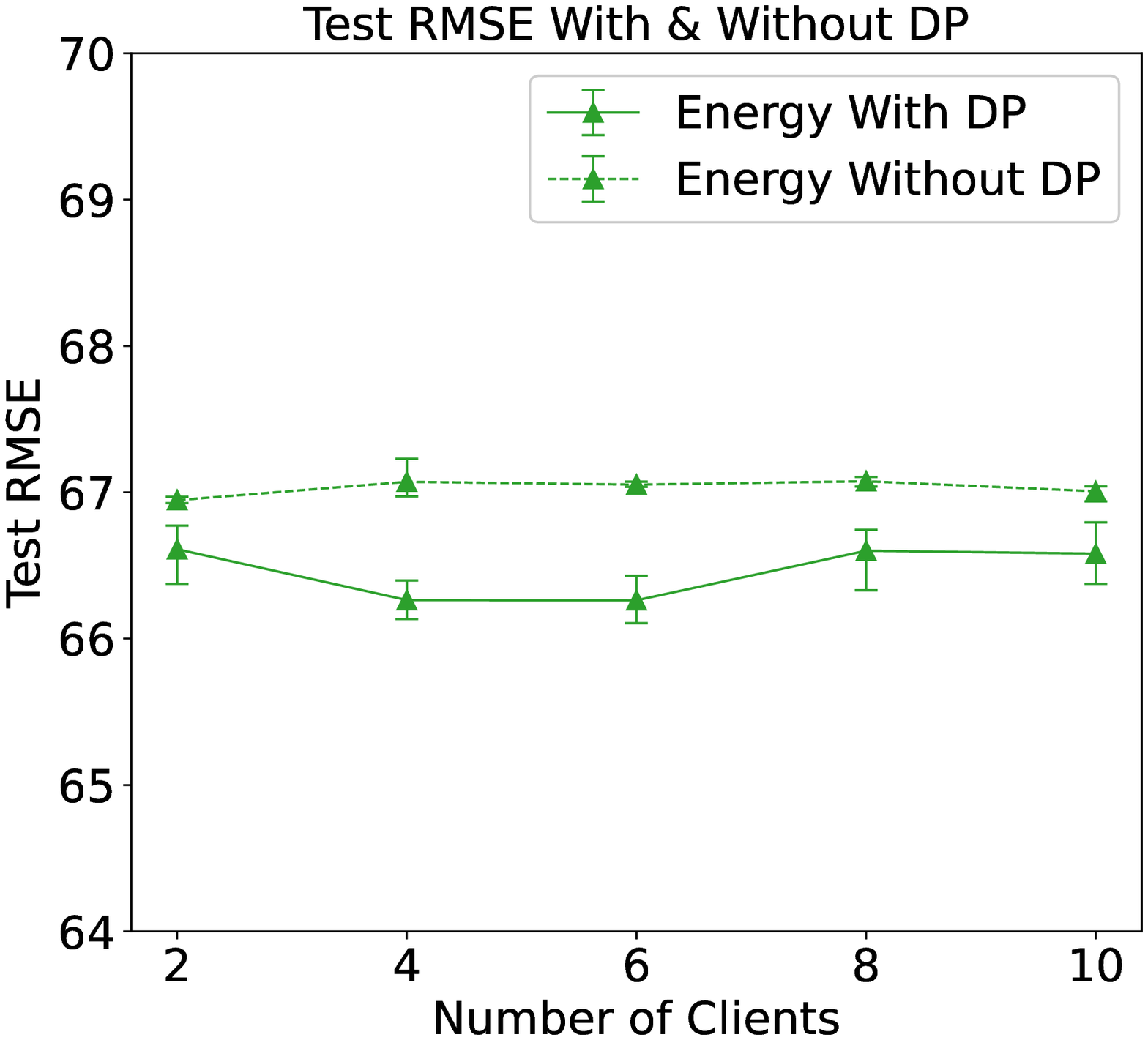}
\end{minipage}
} \\
\caption{The test accuracy and RMSE with and without applying DP over different numbers of participating clients, where (a) is the test accuracy on Credit card and Bank marketing datasets, and (b) is the test RMSE on the regression dataset.}
\label{fig:performclients}
\end{minipage}
\end{figure}

Fig. \ref{fig:performdepth} shows the performance when varying the number of maximum depth of the boosting tree. For classification tasks, the accuracy is about 82$\%$ and 90$\%$ for Credit card and Bank marketing datasets, respectively. And test accuracy of these two datasets remains almost constant as the number of maximum depth increases, except that the results on the Credit dataset with DP have slight fluctuations when the number of maximum depth is 5. Moreover, it is easy to find that the performances with DP and without DP are almost the same, which shows that the impact of the partial DP on the model performance is negligible. On the regression task, the average RMSE decreases as the maximum depth increases, since more complex tree structures may have better model performance. Similar to the previous cases, the test performance with DP is better than that without DP because of over-fitting.


\begin{figure}[!t]
\begin{minipage}[t]{1\linewidth}
\centering
\subfigure[Credit card and Bank marketing]{
\begin{minipage}[b]{0.46\textwidth}
\includegraphics[width=1\textwidth]{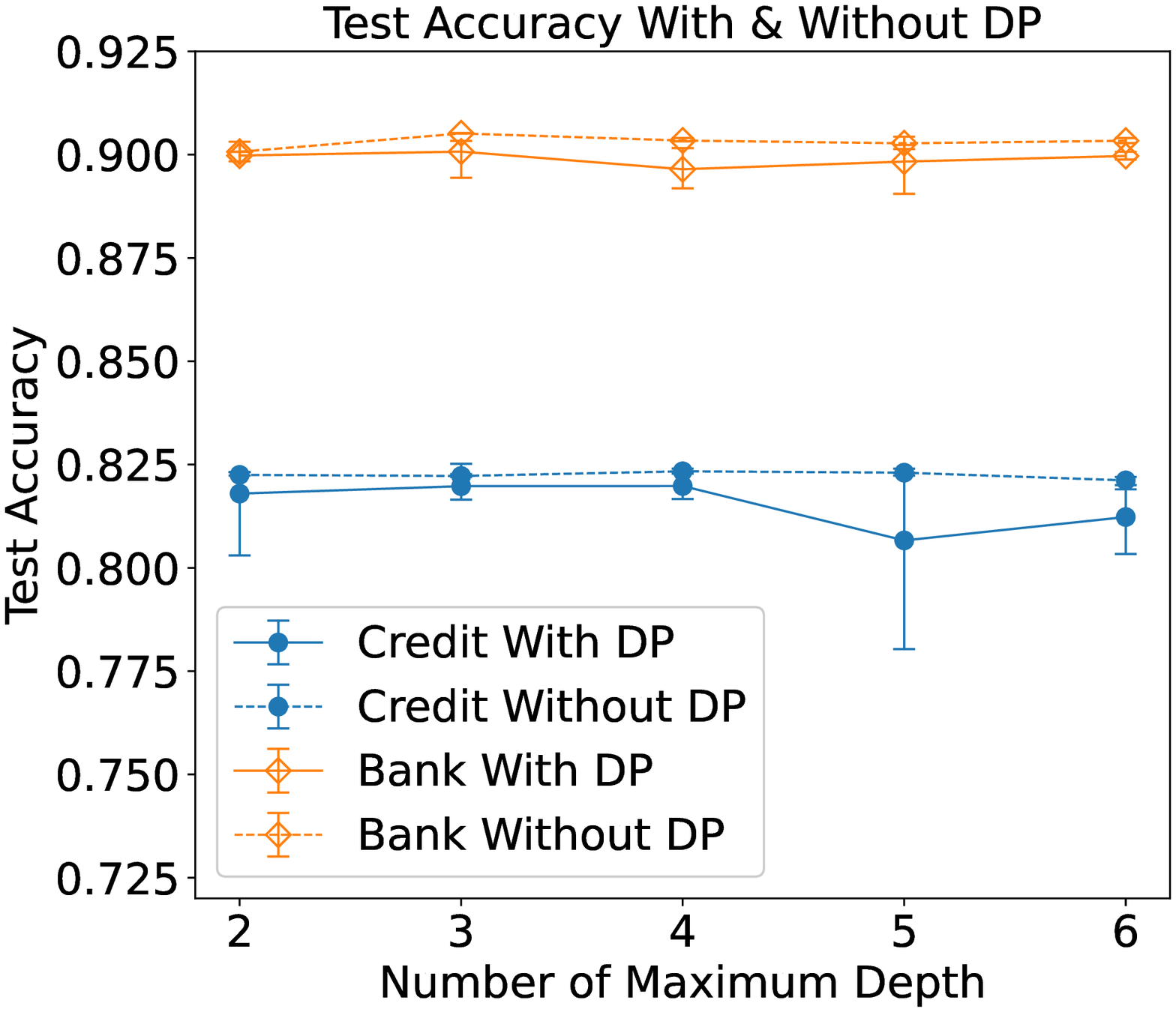}
\end{minipage}
}
\centering
\subfigure[Appliances energy prediction]{
\begin{minipage}[b]{0.46\textwidth}
\includegraphics[width=1\textwidth]{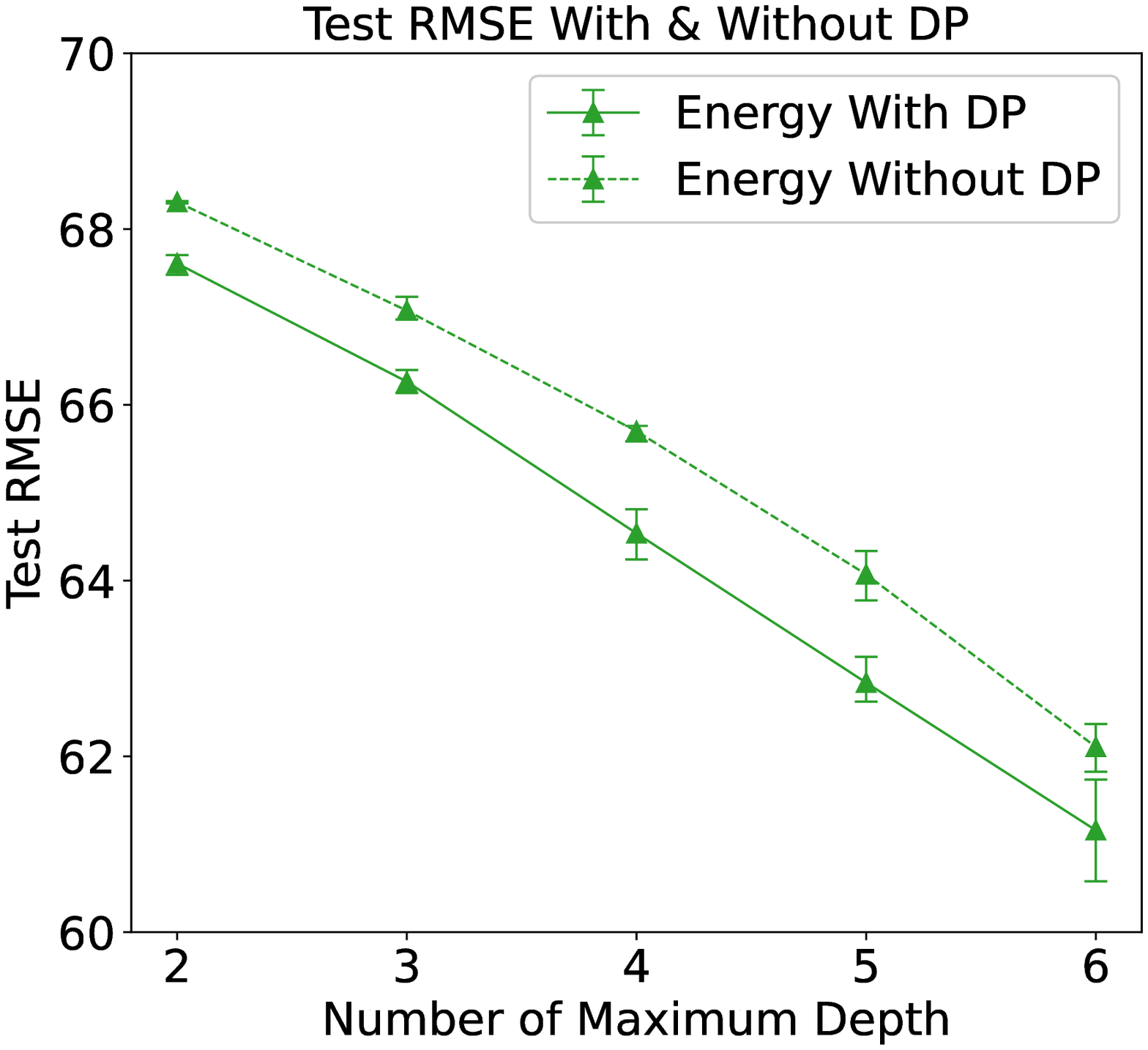}
\end{minipage}
} \\
\caption{The test accuracy and RMSE with and without applying DP over different maximum depths of a boosting tree, where (a) is the test accuracy on the Credit card and Bank marketing datasets, and (b) is the test RMSE on the regression dataset.}
\label{fig:performdepth}
\end{minipage}
\end{figure}

The performance change with the number of boosting trees is shown in Fig. \ref{fig:performtrees}. The results imply that only slight performance changes are observed when the number of boosting trees changes. 
On the Appliances energy prediction dataset, the RMSE slightly decreases with the increase in the number of boosting trees and the RMSE values with and without DP are nearly the same.

\begin{figure}[!t]
\begin{minipage}[t]{1\linewidth}
\centering
\subfigure[Credit card and Bank marketing]{
\begin{minipage}[b]{0.46\textwidth}
\includegraphics[width=1\textwidth]{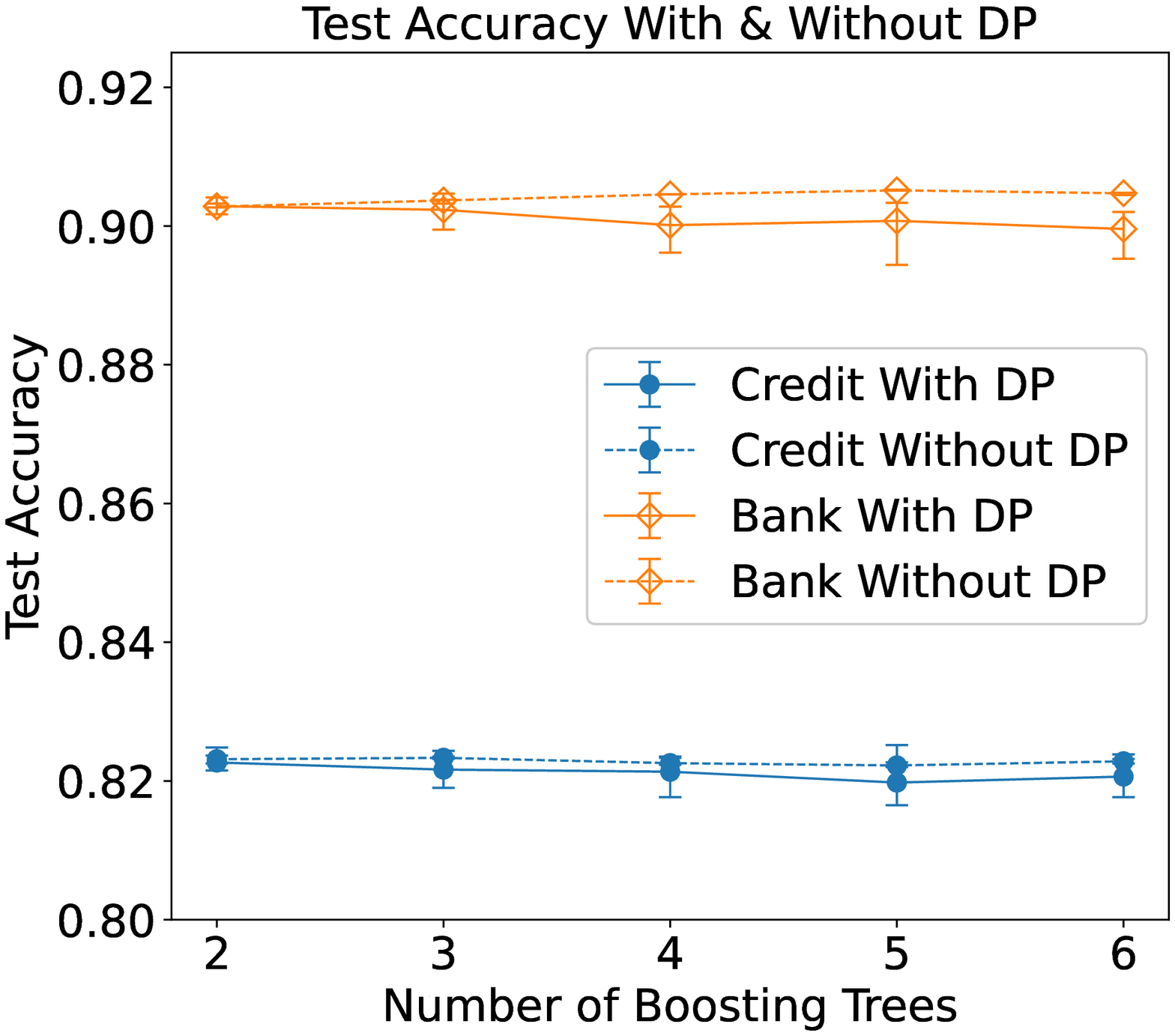}
\end{minipage}
}
\centering
\subfigure[ Appliances energy prediction]{
\begin{minipage}[b]{0.46\textwidth}
\includegraphics[width=1\textwidth]{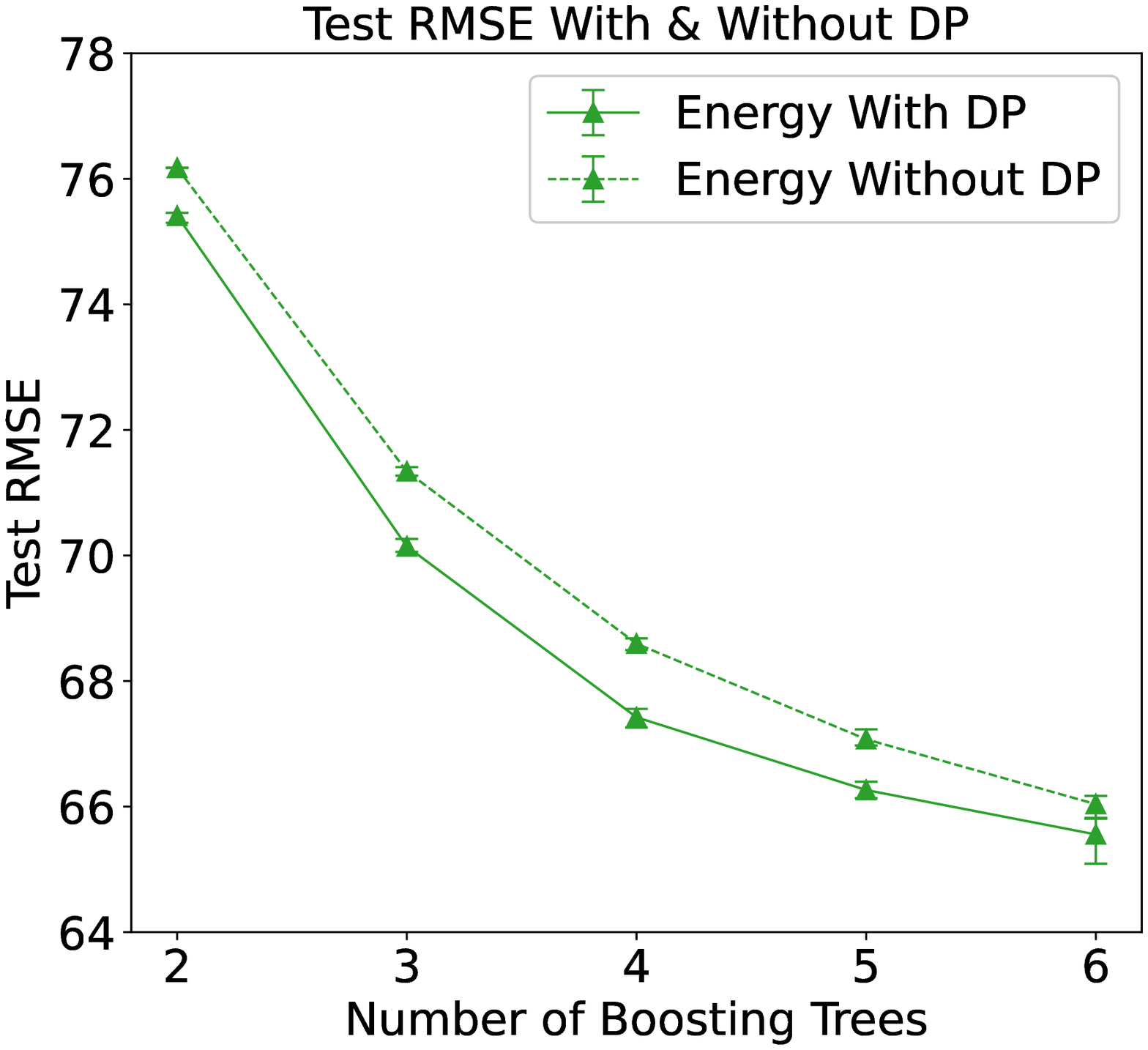}
\end{minipage}
} \\
\caption{The test accuracy and RMSE with and without applying DP over different numbers of boosting trees, where (a) is the test accuracy on the Credit card and Bank marketing datasets, and (b) is the test RMSE of the Appliances energy prediction dataset.}
\label{fig:performtrees}
\end{minipage}
\end{figure}

The test results on the three datasets over different $\epsilon$ values are shown in Fig.\ref{fig:performanceeps}. And it is surprising to see that the test performances on all three datasets have no clear changes with the decrease of the $\epsilon$ value. This is because the partial DP is applied only on part of predicted labels in the source client, which makes it very unlikely to influence the node split at the next level of node split. In addition, the split clients use de-noised leaf values for prediction during the test, which further reduces the performance biases caused by the proposed partial DP method. 

\begin{figure}[!t]
\begin{minipage}[t]{1\linewidth}
\centering
\subfigure[Test accuracy with different $\epsilon$]{
\begin{minipage}[b]{0.46\textwidth}
\includegraphics[width=1\textwidth]{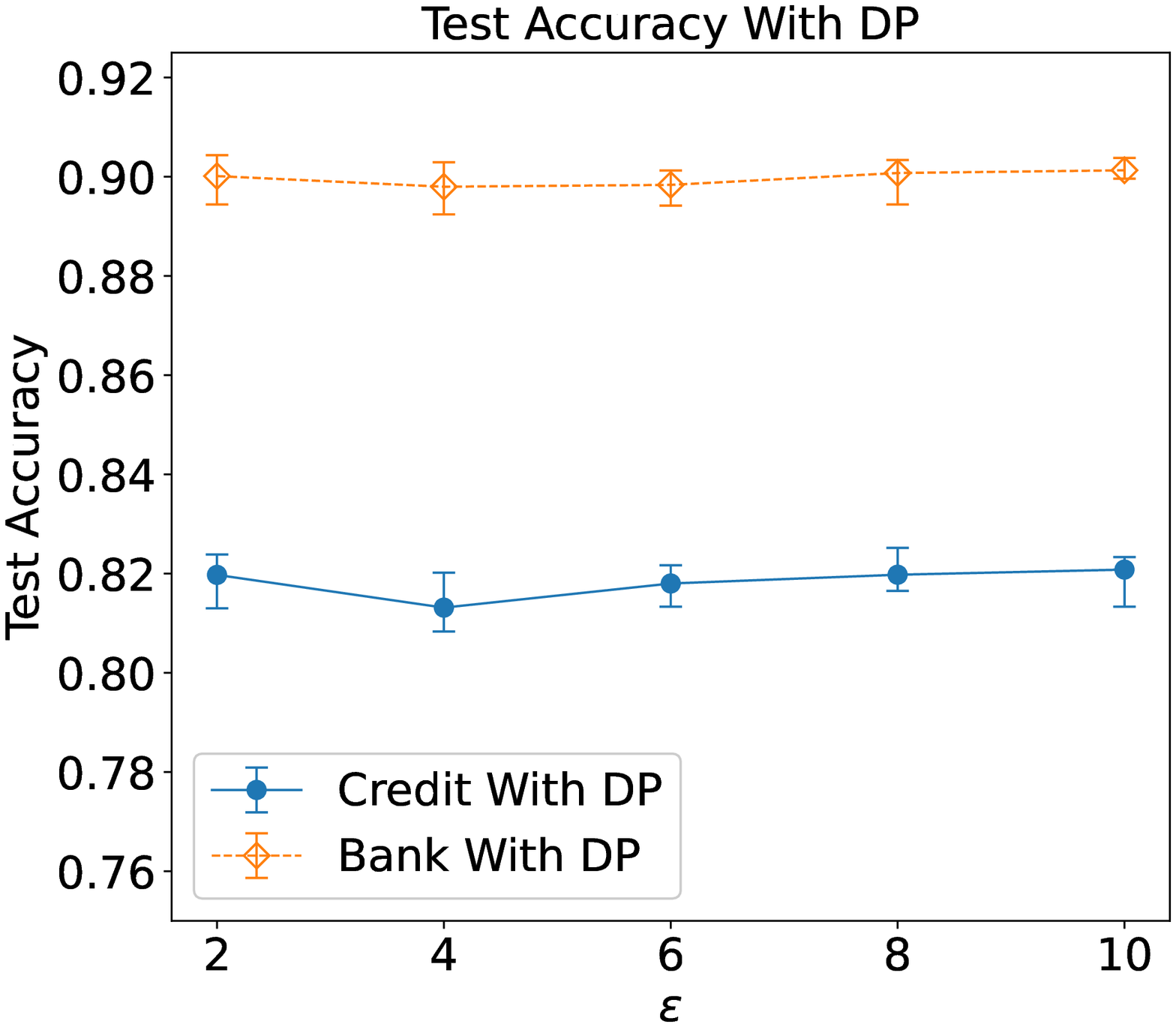}
\end{minipage}
}
\centering
\subfigure[Test RMSE with different $\epsilon$]{
\begin{minipage}[b]{0.46\textwidth}
\includegraphics[width=1\textwidth]{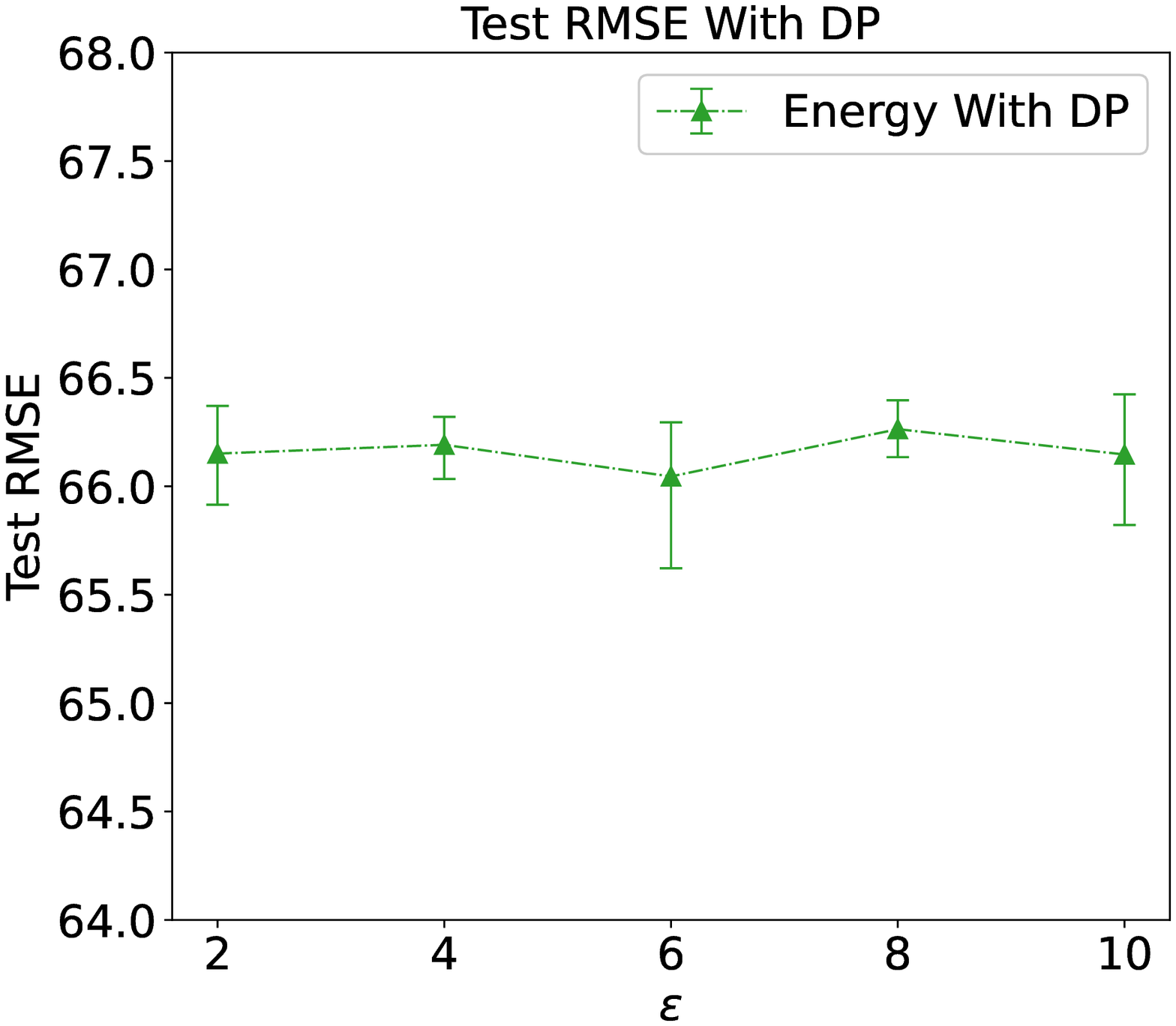}
\end{minipage}
} \\
\caption{The test performance over different $\epsilon$ values, where (a) is the test accuracy on the Credit card and Bank marketing datasets, and (b) is the test RMSE on Appliances energy prediction dataset.}
\label{fig:performanceeps}
\end{minipage}
\end{figure}

\subsection{Evaluation on training time}
In this section, we empirically analyze the training time affected by the number of participating clients, the number of maximum depth and the number of boosting trees. 
All the experiments are run on Intel Core i7-8700 CPU and the Paillier encryption system is implemented with the package in \cite{PythonPaillier}. Just like what we did in the previous set of experiments, the mean values together with the best and worst performances out of five independent runs are included in the results of the ablation studies.

Fig. \ref{fig:timeclients} shows the training time when varying the number of clients with and without encryption. In general, it consumes the least training time on the the bank marketing dataset since it has the smallest amount of data. When no parameter encryption is adopted, the training time of the three datasets almost keeps constant over different numbers of clients, which is consistent with the results reported in \cite{cheng2021secureboost,tian2020federboost}. The reason is that the data samples are partitioned across data features in VFL and setting different numbers of clients does not change the entire data entries. On the contrary, the training time with Paillier encryption increases linearly with the increase in the client numbers, since the amount of the encryption times is proportional to the number of clients in our proposed PIVODL algorithm.

\begin{figure}[!t]
\begin{minipage}[t]{1\linewidth}
\centering
\subfigure[Training time with Paillier encryption]{
\begin{minipage}[b]{0.46\textwidth}
\includegraphics[width=1\textwidth]{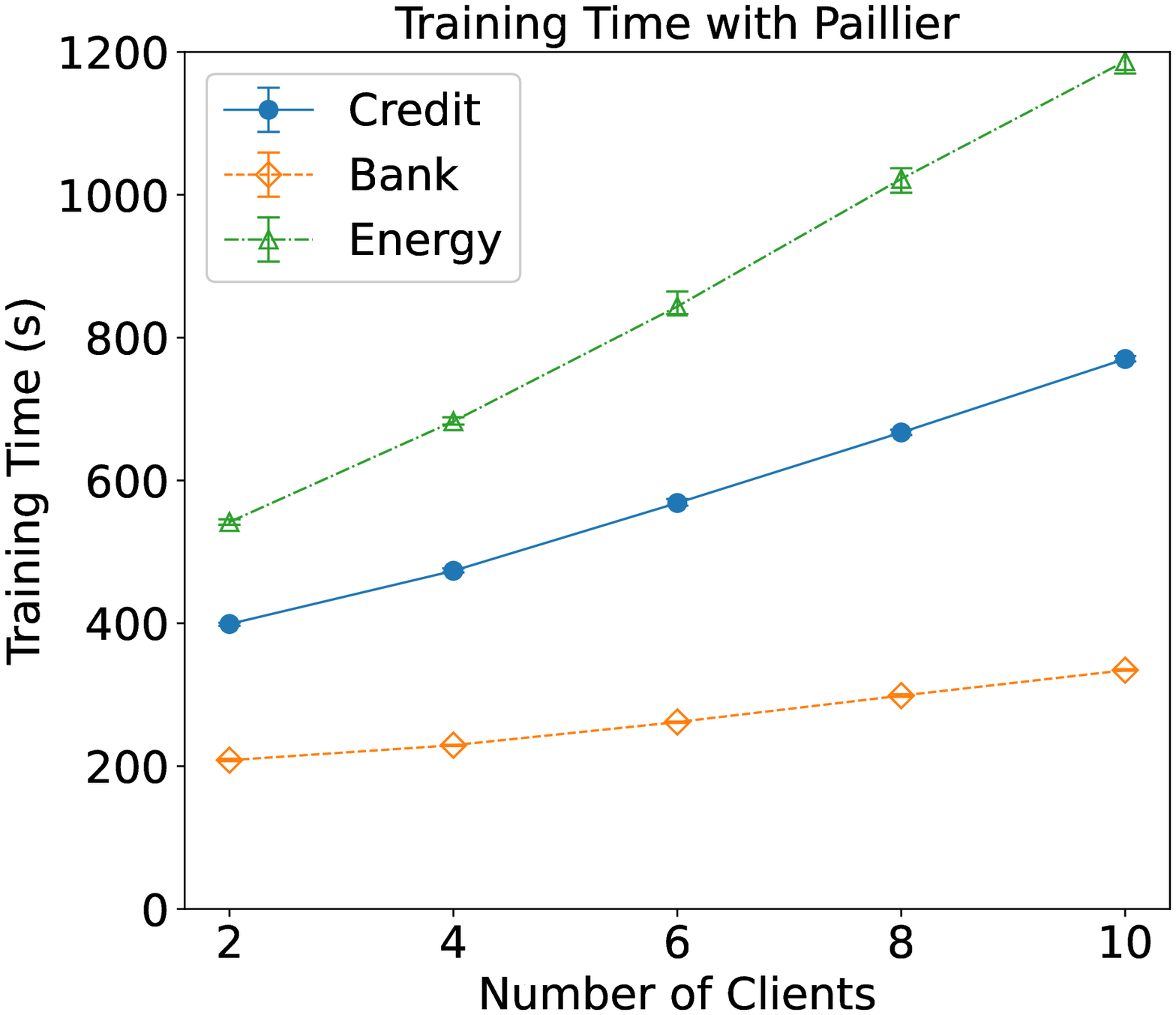}
\end{minipage}
}
\centering
\subfigure[Training time without encryption]{
\begin{minipage}[b]{0.46\textwidth}
\includegraphics[width=1\textwidth]{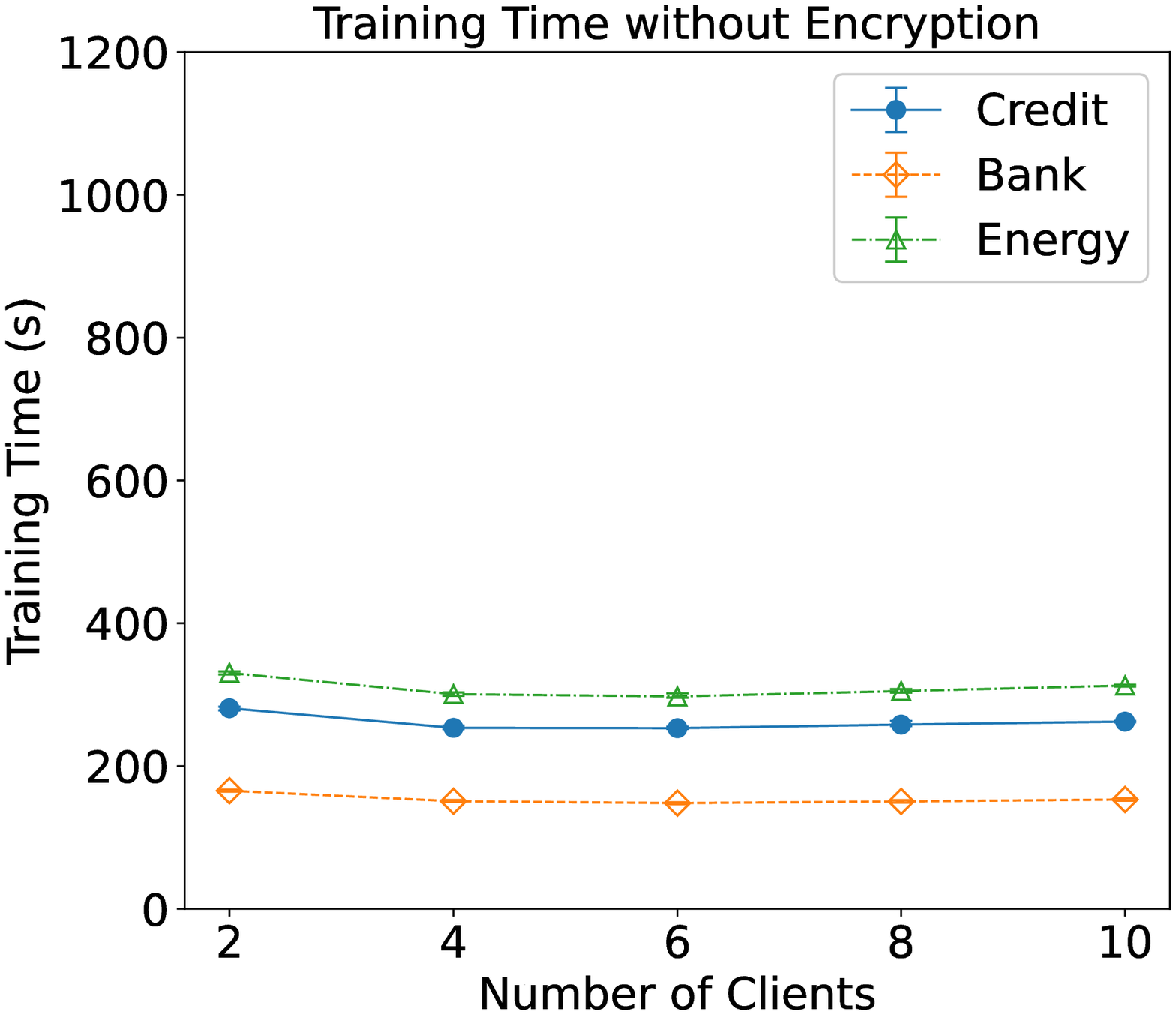}
\end{minipage}
} \\
\caption{The total training time over different numbers of participating clients, where (a) is the training time with Paillier encryption and (b) is the training time without encryption.}
\label{fig:timeclients}
\end{minipage}
\end{figure}

The training times over different maximum depths of a boosting tree are shown in Fig. \ref{fig:timedepth}, and similarly, the cases with and without Paillier encryption are presented for comparison. The training time grows exponentially with the increasing maximum depth, and this phenomenon becomes more obvious with Pailler encryption is adopted. This makes sense since the computational complexity of one decision is $O(2^{n})$, where $n$ is the maximum depth of the tree. More specifically, the training time on the Appliance energy prediction dataset for six different depths is about 2 and 5 minutes without and with encryption, respectively.

\begin{figure}[!t]
\begin{minipage}[t]{1\linewidth}
\centering
\subfigure[Training time with Paillier encryption]{
\begin{minipage}[b]{0.46\textwidth}
\includegraphics[width=1\textwidth]{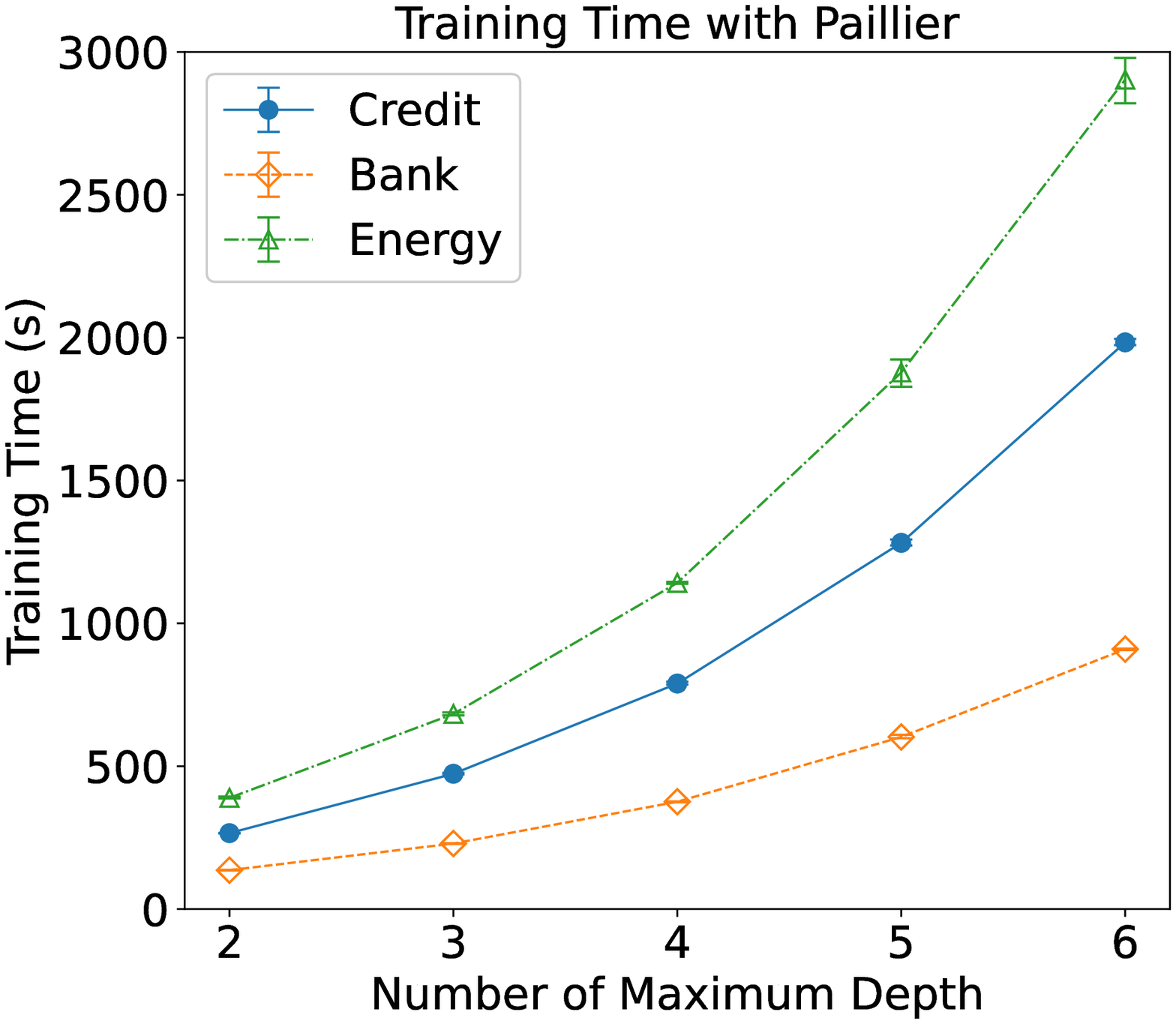}
\end{minipage}
}
\centering
\subfigure[Training time without encryption]{
\begin{minipage}[b]{0.46\textwidth}
\includegraphics[width=1\textwidth]{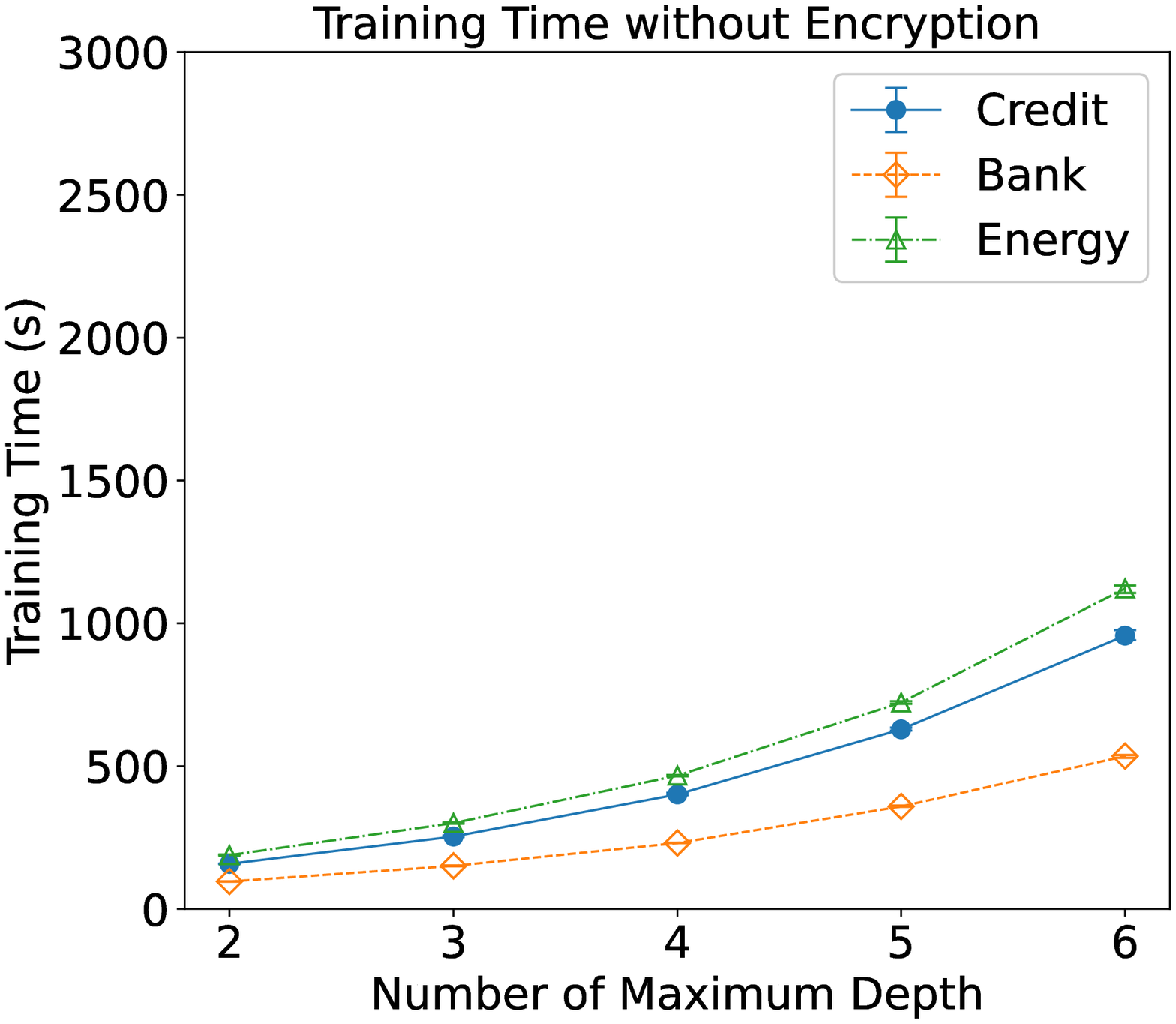}
\end{minipage}
} \\
\caption{The total training time over different maximum depths of a boosting tree, where (a) is the training time with Paillier encryption and (b) is the training time without encryption.}
\label{fig:timedepth}
\end{minipage}
\end{figure}

Fig. \ref{fig:timetrees} indicates the training time consumption over different numbers of boosting trees. The training time increases linearly on the three datasets with increase in the number of boosting trees. Similar to the previous results, the time consumption when using Paillier encryption is much larger than that without encryption. For example, the runtime on the energy dataset for six boosting trees is approximately 800 seconds with encryption, while it takes only about 350 seconds without encryption. 

\begin{figure}[!t]
\begin{minipage}[t]{1\linewidth}
\centering
\subfigure[Training time with Paillier encryption]{
\begin{minipage}[b]{0.46\textwidth}
\includegraphics[width=1\textwidth]{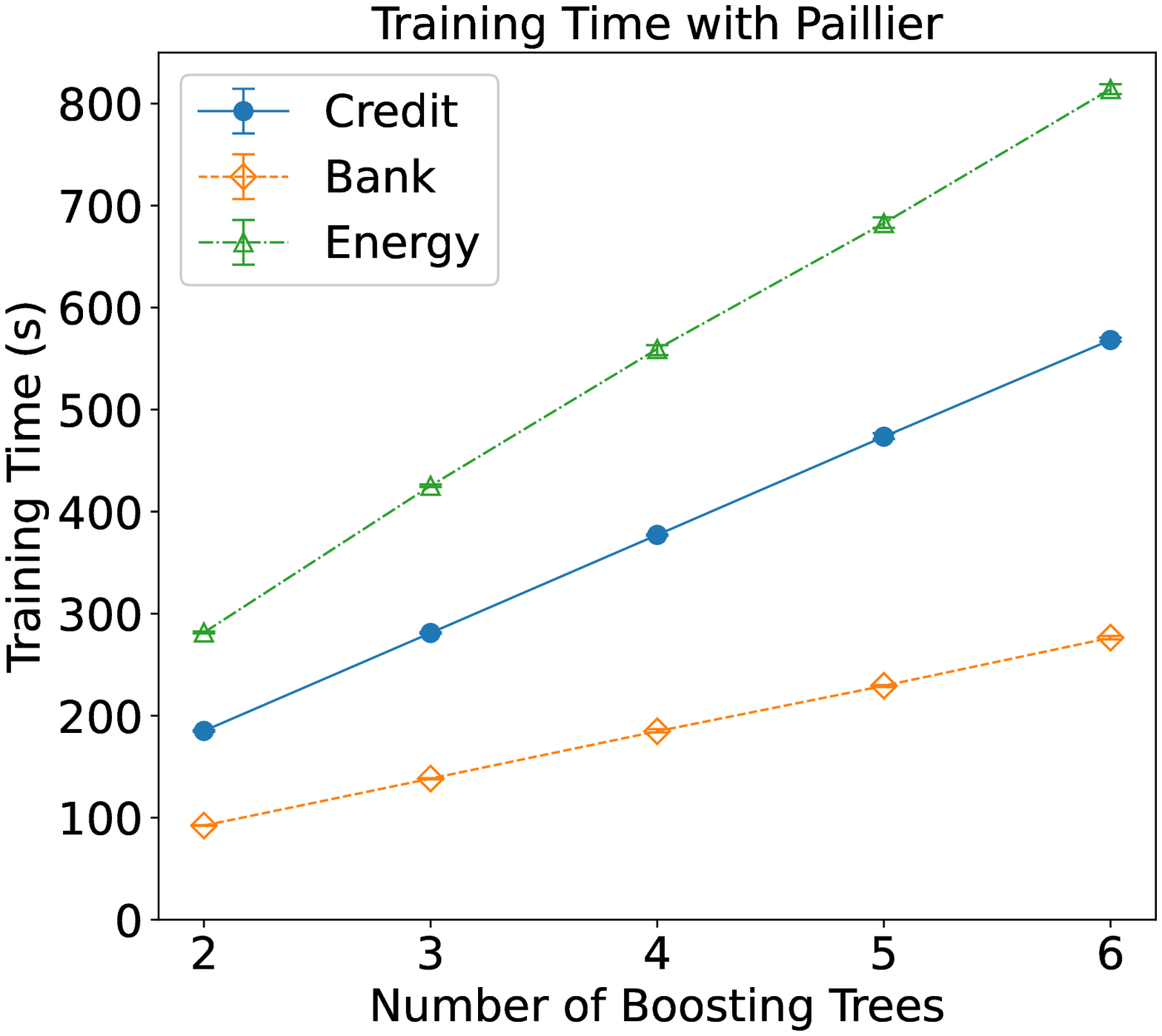}
\end{minipage}
}
\centering
\subfigure[Training time without encryption]{
\begin{minipage}[b]{0.46\textwidth}
\includegraphics[width=1\textwidth]{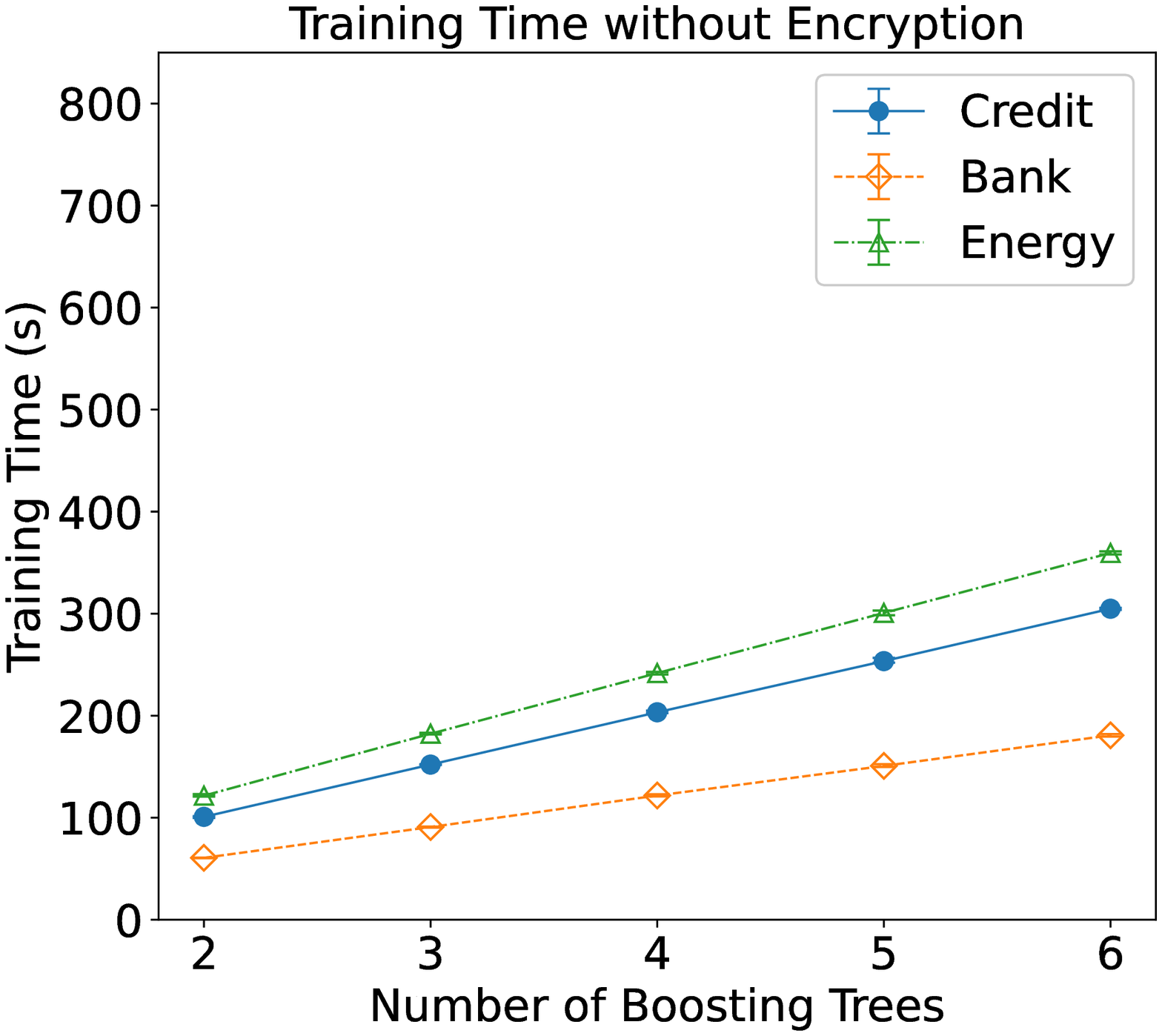}
\end{minipage}
} \\
\caption{The total training time over different number of boosting trees, where (a) is the training time with Paillier encryption and (b) is the training time without encryption.}
\label{fig:timetrees}
\end{minipage}
\end{figure}

\subsection{Evaluation on communication cost}
In this section, we empirically analyze the communication costs (in MB) affected by the number of participating clients, the number of the maximum depth and the number of boosting trees. 

Communication costs are measured by varying the number of clients for both cases with and without using Paillier encryption.
As shown in Fig. \ref{fig:costclient}, the total communication costs for the three datasets increase with the number of clients. More specifically, the communication costs go up dramatically from two clients to four clients, and then increase relatively slightly from four clients to ten clients. Different from the results on the training time, the communication costs with Paillier encryption increase but not as much as the computation time compared to the cases without encryption. This can be attributed to the fact that it is unnecessary to transmit the ciphertexts of the gradients and Hessian values of each data sample for all possible node splits, and only the ciphertexts of the intermediate summation values are required for transmission (see line 27 in Algorithm \ref{alg:securesplit}). 

\begin{figure}[!t]
\begin{minipage}[t]{1\linewidth}
\centering
\subfigure[Communication costs with Paillier encryption]{
\begin{minipage}[b]{0.46\textwidth}
\includegraphics[width=1\textwidth]{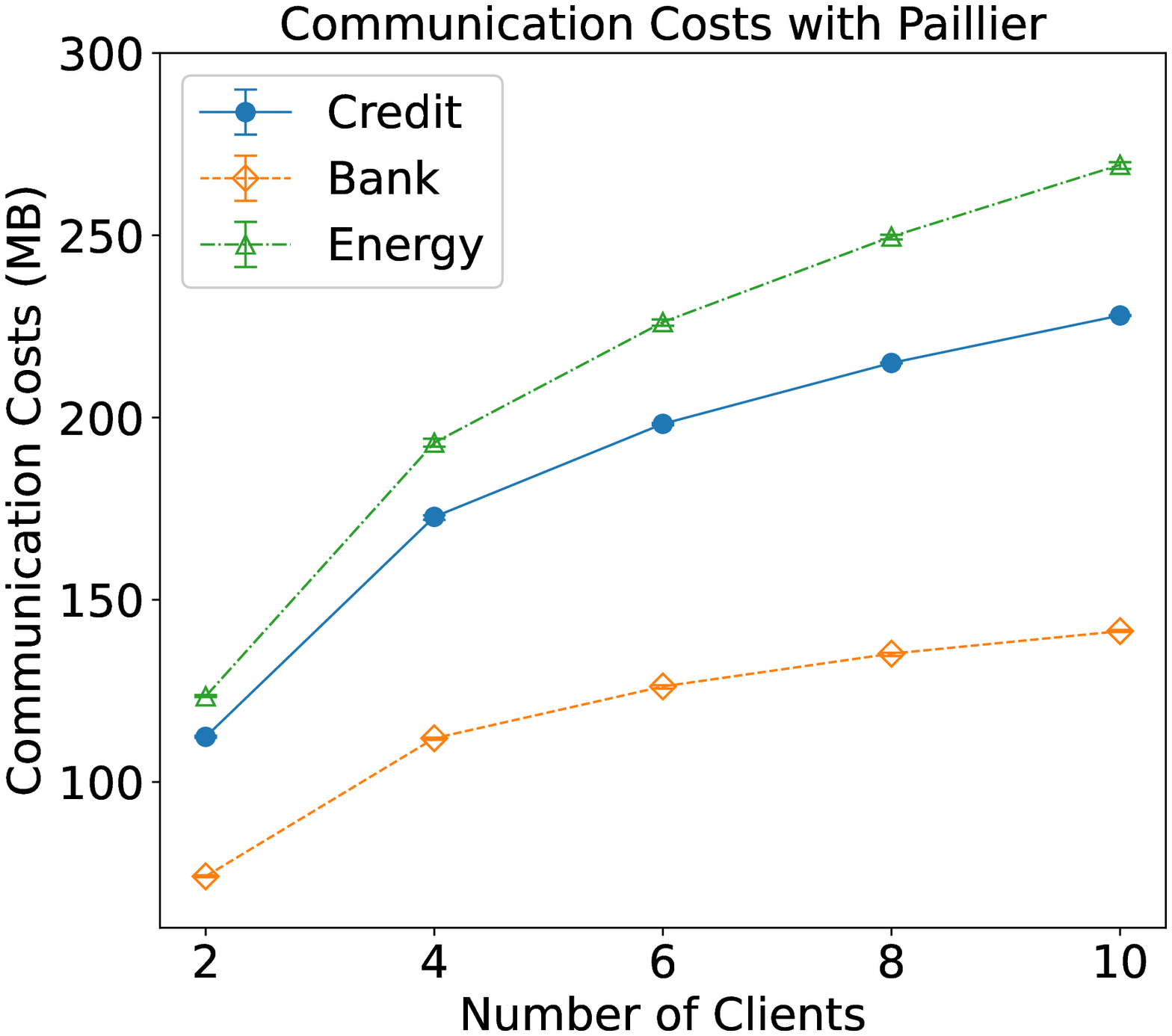}
\end{minipage}
}
\centering
\subfigure[Communication costs without encryption]{
\begin{minipage}[b]{0.46\textwidth}
\includegraphics[width=1\textwidth]{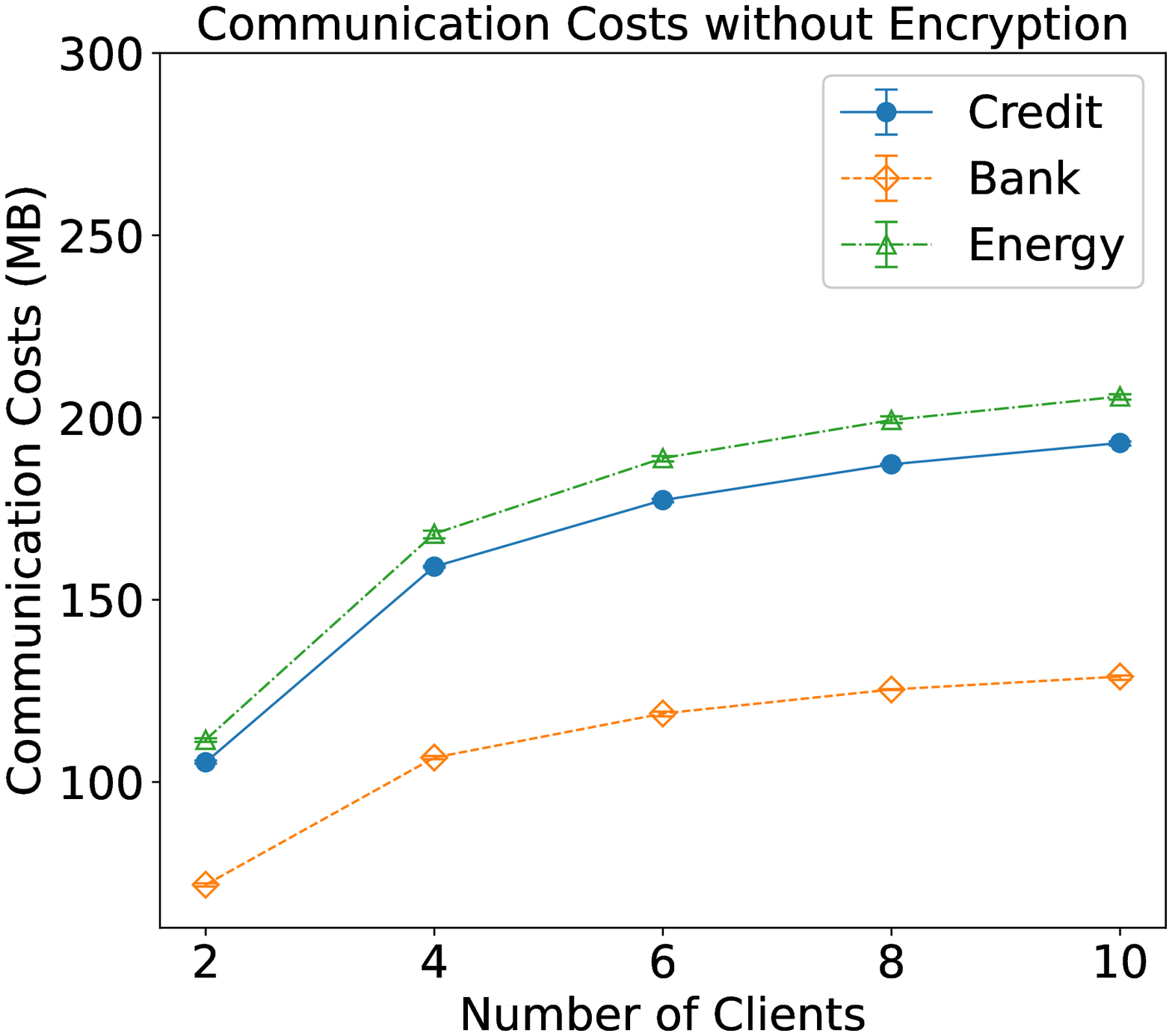}
\end{minipage}
} \\
\caption{The total communication costs over different number of clients, (a) with Paillier encryption, and (b) without encryption.}
\label{fig:costclient}
\end{minipage}
\end{figure}

Fig. \ref{fig:costdepth} shows the communication costs over different maximum depths with and without Paillier encryption. It is apparent to see that the communication costs are proportional to the maximum depth in both encryption and non
-encryption scenarios. The communication costs are similar for the three datasets when the maximum depth is two, and they grow almost linearly with increase in the tree depth. The communication costs of the Appliances energy prediction dataset are 290MB without encryption and 425MB with encryption for a depth of six, which is the largest among all the three datasets. 

\begin{figure}[!t]
\begin{minipage}[t]{1\linewidth}
\centering
\subfigure[Communication costs with Paillier encryption]{
\begin{minipage}[b]{0.46\textwidth}
\includegraphics[width=1\textwidth]{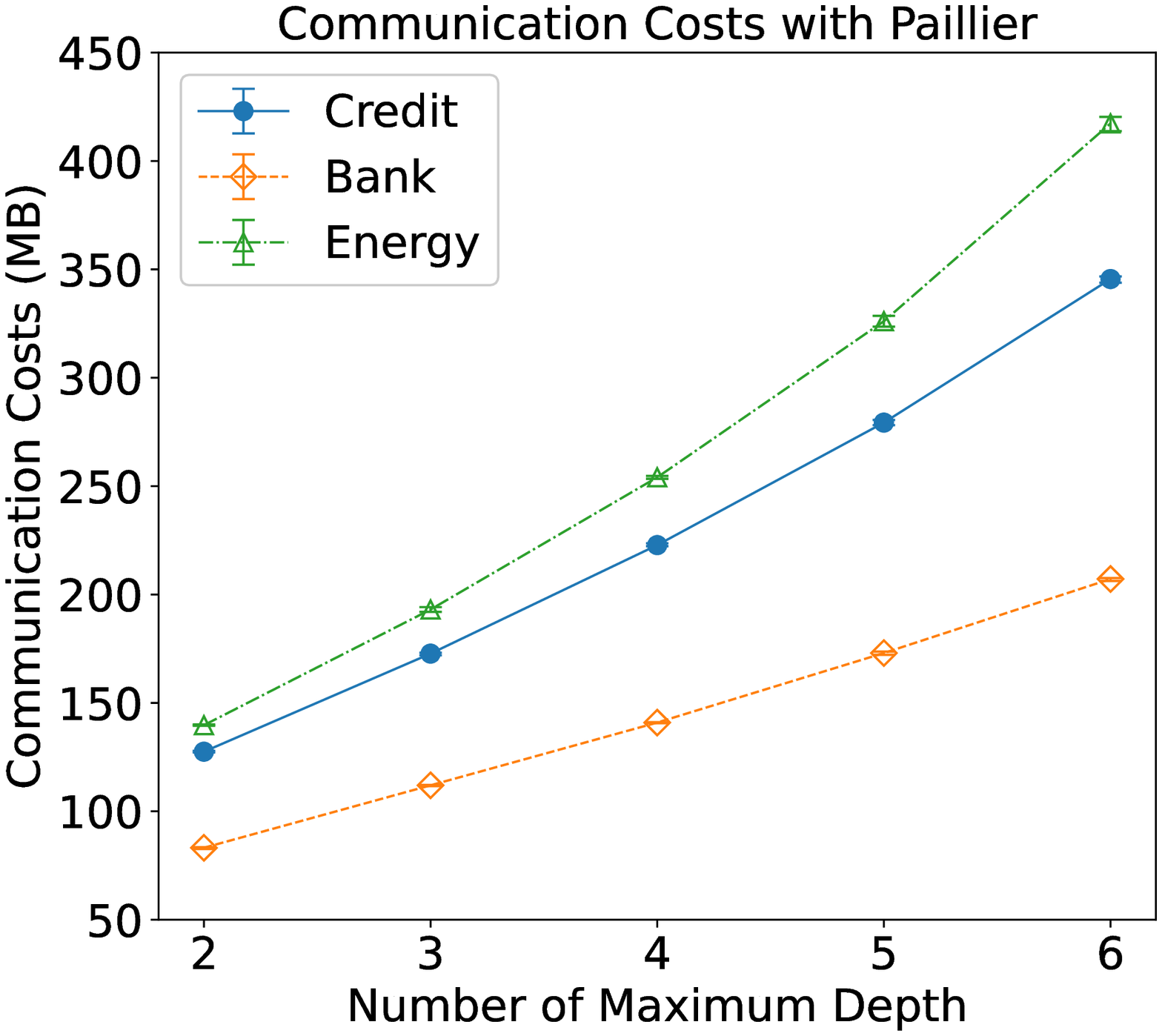}
\end{minipage}
}
\centering
\subfigure[Communication costs without encryption]{
\begin{minipage}[b]{0.46\textwidth}
\includegraphics[width=1\textwidth]{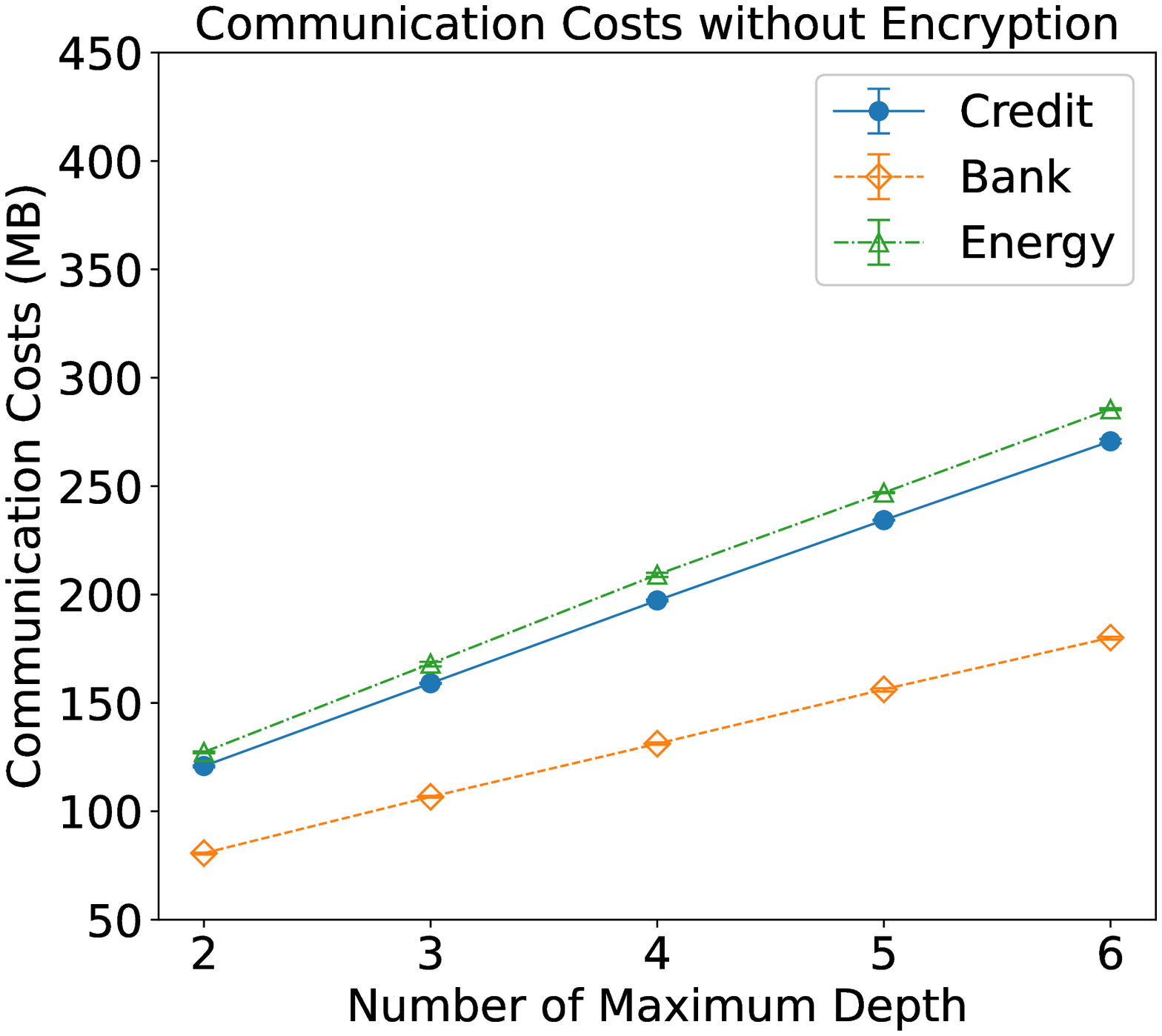}
\end{minipage}
} \\
\caption{The total communication costs over different maximum depths of a boosting tree, where (a) is the communication costs with Paillier encryption and (b) is the communication costs without encryption}
\label{fig:costdepth}
\end{minipage}
\end{figure}

The communication costs related to the number of boosting trees are shown in Fig .\ref{fig:costtree}. Similar to the previous results, the communication costs increase linearly with the number of trees, since the model size of XGBoost is proportional to the number of trees. And it is clear to see that the impact of encryption on communication costs is not as significant as that on the training time. The encryption brings a maximum of 25MB extra communication costs for the model with six boosting trees on the Appliance energy prediction dataset.

\begin{figure}[!t]
\begin{minipage}[t]{1\linewidth}
\centering
\subfigure[Communication costs with Paillier encryption]{
\begin{minipage}[b]{0.46\textwidth}
\includegraphics[width=1\textwidth]{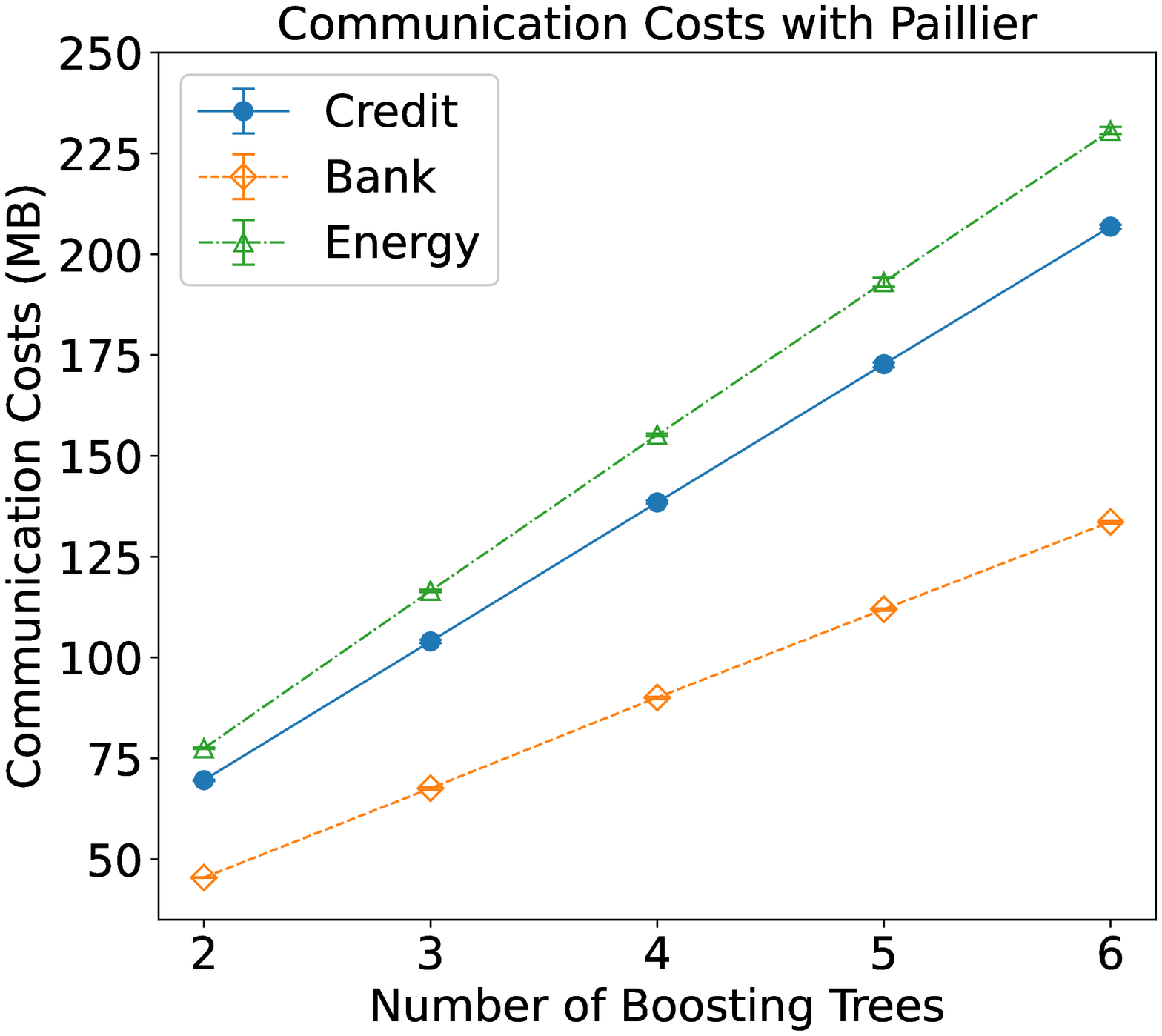}
\end{minipage}
}
\centering
\subfigure[Communication costs without encryption]{
\begin{minipage}[b]{0.46\textwidth}
\includegraphics[width=1\textwidth]{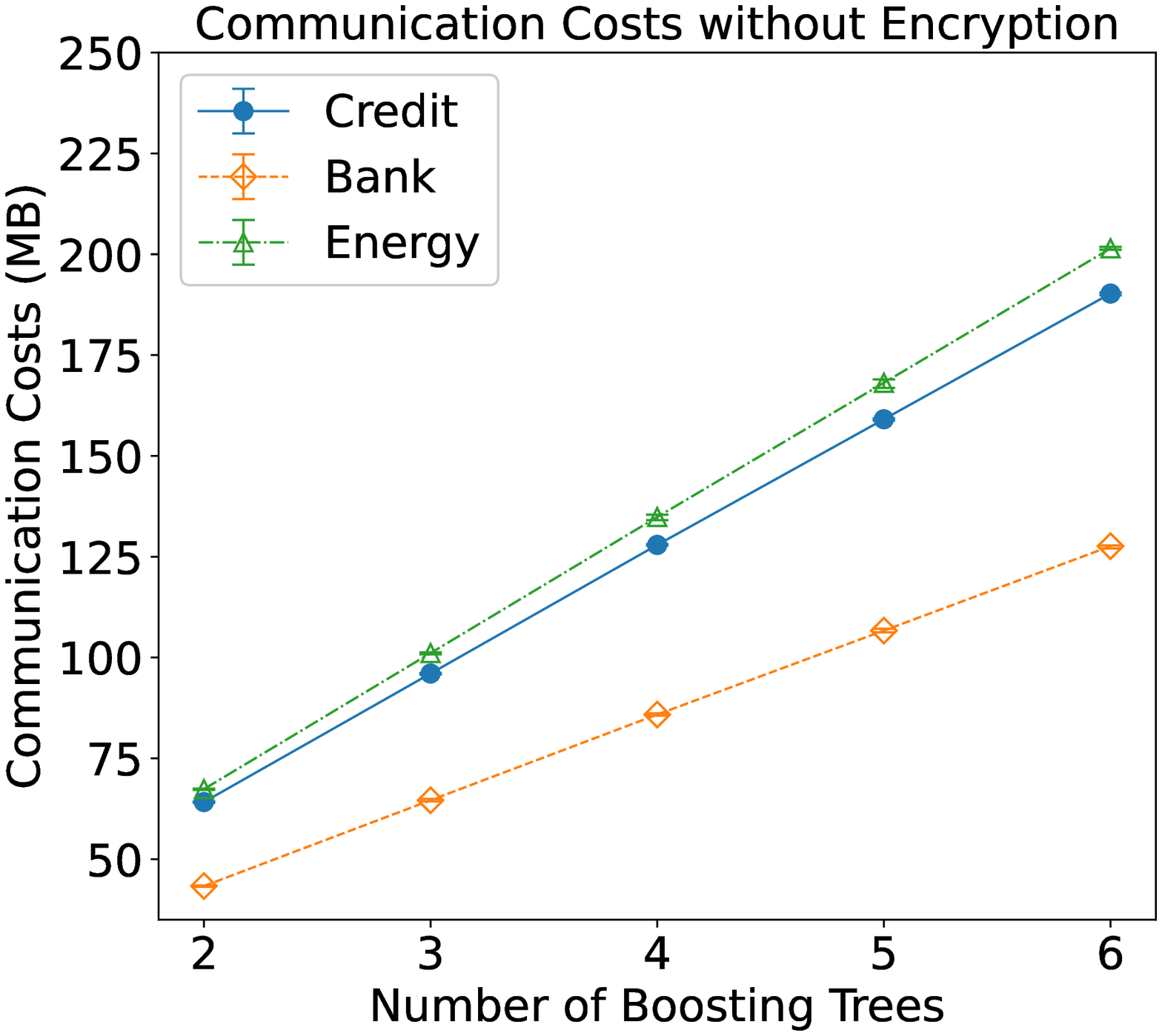}
\end{minipage}
} \\
\caption{The total communication costs over different numbers of boosting trees, where (a) is the communication costs with Paillier encryption and (b) is the communication costs without encryption.}
\label{fig:costtree}
\end{minipage}
\end{figure}

\subsection{Inference of predicted labels}
In this section, we assume that each client is curious but honest and attempts to deduce the true label information of other clients through the available leaf weights. The attack strategy is, for instance, to ensemble all the leaf values and guess a label if one client knows several leaf values of one specific data sample. By comparing the 'guessed' labels with the true labels from other clients, it is easy to calculate the guess accuracy of each client. Note that for a binary classification problem, the data labels are either 0 or 1 and the probability of a correct random guess is 50\%. Therefore, if the guess accuracy is equal to or less than 50\%, the label privacy is not revealed.

\begin{table}[]
\caption{Guess accuracy on different datasets}
\begin{tabular}{|p{26pt}|p{24pt}|p{24pt}|p{24pt}|p{24pt}|p{24pt}|p{24pt}|}
\hline
 \begin{tabular}[c]{@{}l@{}}Dataset\end{tabular} & \begin{tabular}[c]{@{}l@{}}Guess \\ accuracy,\\  no DP\end{tabular} & \begin{tabular}[c]{@{}l@{}}Guess \\ accuracy, \\ $\epsilon=2$\end{tabular} & \begin{tabular}[c]{@{}l@{}}Guess \\ accuracy, \\ $\epsilon=4$\end{tabular} & \begin{tabular}[c]{@{}l@{}}Guess \\ accuracy, \\ $\epsilon=6$\end{tabular} & \begin{tabular}[c]{@{}l@{}}Guess \\ accuracy, \\ $\epsilon=8$\end{tabular} & \begin{tabular}[c]{@{}l@{}}Guess \\ accuracy, \\ $\epsilon=10$\end{tabular} \\ \hline
\begin{tabular}[c]{@{}l@{}}Credit \\ card\end{tabular}                                                                & 68.75\%                                                             & 30.02\%                                                                 & 23.02\%                                                                 & 34.15\%                                                                 & 27.78\%                                                                 & 40.93\%                                                                  \\ \hline
\begin{tabular}[c]{@{}l@{}}Bank \\ marketing\end{tabular}                                                             & 61.57\%                                                             & 39.57\%                                                                 & 38.85\%                                                                 & 25.85\%                                                                 & 39.73\%                                                                 & 14.45\%                                                                  \\ \hline
\end{tabular}

\label{guessacc}
\end{table}

For the sake of brevity, we calculate the expected guess accuracy of all participating clients with and without partial DP using different $\epsilon$ values.
As shown in Table \ref{guessacc}, the guess accuracy is more than 60\% for both two datasets without using partial DP, which is more accurate than a random guess.
After applying the proposed partial DP algorithm, however, the guess accuracies drop dramatically and are all less than 50\%. Specifically on the Credit card dataset, the guess accuracy decreases to around 27\% for $\epsilon=10$ and 45\% for $\epsilon=4$. On the Bank marketing dataset, the guess accuracy drops to 14.45$\%$ when $\epsilon=10$. Since the noise in the DP follows the Gaussian distribution, it is normal that a smaller $\epsilon$ values will give a higher guess accuracy. But as long as partial DP is used for label privacy protection, the guess accuracy is lower than a random guess, meaning that no label privacy is leaked from the received leaf outputs.

\section{Conclusion and future work}
In this paper, we propose a secure learning system, called PIVODL, for privately training XGBoost decision tree models in a vertical federated learning environment with labels distributed on multiple clients. A secure node split protocol and a privacy-preserving tree training algorithm are proposed by requiring the source and split clients to separately split the data samples and calculate the corresponding impurity scores, effectively defending against differential attacks that may be encountered during the boosting tree node splits. In addition, a partial DP mechanism is adopted to deal with label privacy so that no client is able to guess the correct data labels through the revealed leaf weights.

Our experimental results empirically indicate that the proposed PIVODL framework is able to securely construct ensemble trees with negligible performance degradation. We show the test performance of the resulting decision trees is relatively insensitive to different $\epsilon$ values of the introduced DP, since the proposed algorithm adds Gaussian noise only to the leaf weights that are sent to the source clients. Moreover, the partial DP can effectively prevent labels of the data from being revealed through the received leaf weights.

Although the proposed PIVODL system shows promising performance in VFL, improving training efficiency remains challenging. The encryption in PIVODL still takes a large proportion of both learning time and communication costs. Therefore, our future work will be dedicated to developing a light weighted and efficient secure XGBoost system for VFL with labels distributed on multiple devices.





\ifCLASSOPTIONcaptionsoff
  \newpage
\fi



%



{\footnotesize\bibliography{references}
\bibliographystyle{ieeetr}}

%








\end{document}